%% file: root.tex
\documentclass[10pt,journal,compsoc]{IEEEtran}

\usepackage{amsmath}
\usepackage{lmodern}
\usepackage{multirow}
\usepackage[utf8]{inputenc}
\usepackage[T1]{fontenc}

\usepackage{amssymb}
\usepackage{amsthm}
\usepackage{mathrsfs}
\usepackage{amsmath}
\usepackage{amstext}
\usepackage{caption}
\usepackage{framed}
\usepackage{units}
\usepackage{mathtools}

\usepackage{subfigure,graphicx} 
\usepackage{epstopdf}

\usepackage{booktabs} 
\usepackage{array} 
\usepackage{paralist} 
\usepackage{verbatim} 
\usepackage{algorithm2e}
\RestyleAlgo{ruled}

\newtheorem{theorem}{Theorem}
\newtheorem{lemma}[theorem]{Lemma}
\newtheorem{corollary}{Corollary}

\begin{document}
\title{FEC for Lower In-Order Delivery Delay in Packet Networks}
\author{Mohammad Karzand, Douglas J. Leith, Jason Cloud, Muriel Medard
\IEEEcompsocitemizethanks{
\IEEEcompsocthanksitem{M.Karzand and D.J.Leith are with Trinity College Dublin, Ireland}
\IEEEcompsocthanksitem{J.Cloud and M.Medard are with Massachusetts Institute of Technology, MA, USA}
}
\thanks{Work by MK and DL supported by Science Foundation Ireland grants 13/IF/I2781 and 11/PI/11771.}
}

\maketitle

\begin{abstract}
We consider use of Forward Error Correction (FEC) to reduce the in-order delivery delay over packet erasure channels. We propose a class of streaming codes that is capacity achieving and provides a superior throughput-delay trade-off compared to block codes by introducing flexibility in where and when redundancy is placed. This flexibility results in significantly lower in-order delay for a given throughput for a wide range of network scenarios.  Furthermore, a major contribution of this paper is the combination of queuing and coding theory to analyze the code's performance. Finally, we present simulation and experimental results illustrating the code's benefits.
\end{abstract}
\vspace{-10pt}
\section{Introduction}
\label{sec:introduction}
\input{introduction.tex}
\vspace{-5pt}
\section{Related Work}
\label{sec:blockcodes}
\input{blockcodes.tex}
\vspace{-5pt}
\section{Low-Delay Coding Over a Stream}
\label{sec:lowdelaycoding}
\input{lowdelaycoding.tex}
\vspace{-5pt}
\section{In-Order Delivery Delay and Throughput}
\label{sec:inorder-delay}
\input{inorderdelay.tex}

\vspace{-5pt}
\section{Decoding Failure Probability}
\label{sec:failure-probability}
\input{failure.tex}

\vspace{-5pt}
\section{Encoding and Decoding Complexity}
\label{sec:complexity}
\input{complexity.tex}

\vspace{-5pt}
\section{Low Delay Code vs. Block Codes}
\label{sec:performancecomp}
\input{sim_exp_results.tex}
\vspace{-5pt}
\section{Conclusions}
\label{sec:conclusions}
\input{conclusion.tex}

\vspace{-5pt}
\bibliographystyle{IEEEtran}
\bibliography{LowDelayCoding}

\begin{IEEEbiography}[{\includegraphics[width=1in,height=1.25in,clip,keepaspectratio]{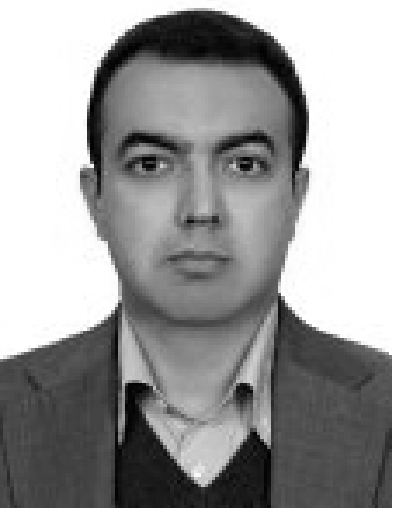}}]{Mohammad Karzand} obtained the Electrical Engineering Degree at the University of Tehran in 2005. Then, he received his M.Sc. (2007) and Ph.D. (2013) from the School of Computer and Communication Sciences in Ecole Polytechnique Fédérale de Lausanne (EPFL). Since 2013, he has been working as a senior research fellow at the Hamilton Institute (NUIM) and Trinity College Dublin in topics related to delay in Communication Networks.
\end{IEEEbiography}

\begin{IEEEbiography}[{\includegraphics[width=1in,height=1.25in,clip,keepaspectratio]{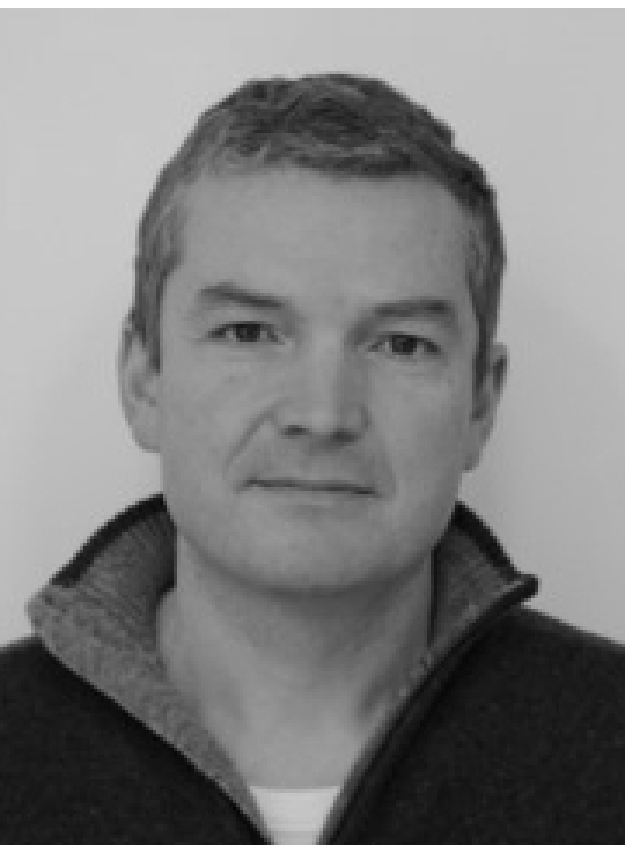}}]{Dough Leith} (SM'01) graduated from the University of Glasgow in 1986 and was awarded his PhD, also from the University of Glasgow, in 1989.  Prof. Leith moved to the National University of Ireland, Maynooth in 2001 to establish the Hamilton Institute (www.hamilton.ie) of which he was founding Director from 2001-2014.  Towards the end of 2014, Prof Leith moved to Trinity College Dublin to take up the Chair in Computer Systems in the School of Computer Science and Statistics.  His current research interests include wireless networks, congestion control,  optimisation and data privacy.
\end{IEEEbiography}

\begin{IEEEbiography}{Jason Cloud} is currently pursuing a Ph.D. degree in Electrical Engineering and Computer Science at the Massachusetts Institute of Technology (MIT). He received a B.S. degree in Engineering with a specialty in Electrical Engineering in 2002 from Colorado School of Mines and a M.S. degree in Electrical Engineering and Computer Science in 2011 from MIT. From 2003 to 2009, he was an officer and developmental engineer in the U.S. Air Force. His research interests include heterogeneous networks, satellite communications, and network coding.
\end{IEEEbiography}
\begin{IEEEbiography}{Muriel Médard} 
\end{IEEEbiography}

%
%
\vspace{-5pt}
\section*{Appendix I: Proof of Theorem  \ref{maingeniedet}}
\input{app1.tex}

\vspace{-5pt}
\section*{Appendix II: Proof of Theorem  \ref{maindelay}}
\input{app4.tex}

\vspace{-5pt}
\section*{Appendix III: Proof of Theorem   \ref{maingeniedetgroup}}
\input{app3.tex}

\vspace{-5pt}
\section*{Appendix IV: Proof of Theorem  \ref{throughputthm}}
\input{app1a.tex}


\vspace{-5pt}
\section*{Appendix V: Proof of Theorems  \ref{theoremfailure}, \ref{theoremfailstr} and Corollary \ref{corollaryfailure}}
\input{app2.tex}

\end{document}

%% file: introduction.tex
In this paper we study the in-order delivery delay of forward error-correction codes in packet switched networks where erasures are due to queue overflow, lossy links, \emph{etc.} (i.e., a packet erasure channel).  In-order delivery delay is an important metric, often of much interest in applications where received information packets are buffered until they can  be delivered in-order.  Any packet loss typically leads to increases in the in-order delivery delay that can adversely affect upper layer performance.  For example, traditional ARQ takes on the order of a round-trip time, or more, to recover from a lost packet. Not only is the lost packet delayed, but all subsequent information packets must be buffered until the loss is corrected (i.e., "head-of-line blocking"). Our interest is in the design and analysis of forward error-correction codes that mitigate this problem and are specifically tailored so that they achieve low in-order delivery delay.

Traditionally, the primary aim when designing error-correction codes is to maximise throughput.  The inherent fact that there is a trade-off between throughput and delay is, of course, recognised in this work (\emph{e.g.,} in the analysis of ARQ schemes with delayed feedback \cite{leith10}), but it is very much a secondary concern.   In contrast, our interest is in sacrificing some throughput in order to achieve much lower in-order delivery delay.  This is motivated by the observation that bandwidth is relatively abundant in modern networks, but delay continues to be a major concern for many applications. The use of some some bandwidth to lower delay is therefore an appealing proposition.   

This brings the trade-off between throughput and delay to the forefront; and, in particular, it raises questions as to whether new codes can be constructed that may achieve a more favourable trade-off than traditional codes.  Consider, for example, a rate $\nicefrac{k}{n}$ systematic block code where a block consists of $k$ information packets followed by $n-k$ coded packets. This is illustrated in Figure \ref{fig:intro}(a).   Suppose that the code is an ideal one in the sense that receipt of any $k$ of the $n$ packets allows all of the $k$ information packets to be reconstructed.  Furthermore, assume that the first information packet is lost. All remaining information packets have to be buffered until the first coded packet is received. At this point, the first information packet can be reconstructed and all of the information packets can be delivered in-order.   The in-order delivery delay is therefore proportional to $k$.  Alternatively, suppose that the $n-k$ coded packets are distributed uniformly among the information packets, rather than all being placed after the $k$ information packets.  To keep the code causal, suppose that each coded packet only protects the preceding information packets in the block.  This is still a block code, but the coded packets and their locations differ from the classical setup (see Figure \ref{fig:intro}(b)).  Assume again that the first information packet is lost. This loss can now be recovered on receipt of the first coded packet resulting in a delay that is now proportional to $\nicefrac{k}{n-k}$ (\emph{i.e.,} this is much lower than $k$ when $n$ is large).  Of course, the impact on throughput due to the fact that each coded packet only protects the preceding information packets requires further analysis.

\begin{figure}
\centering
\subfigure[Systematic Block Code]{
\includegraphics[width=0.4\columnwidth]{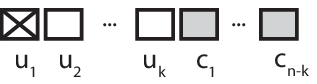}
}\quad
\subfigure[Low Delay Code]{
\includegraphics[width=0.48\columnwidth]{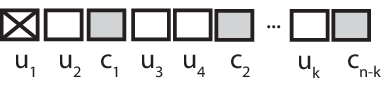}
}
\caption{Example of two codes with different throughput-delay characteristics. Shaded squares indicated coded packets, unshaded indicate information packets.}
\vspace{-15pt}
\label{fig:intro}
\end{figure}

Our main contributions are as follows.  While our proposal falls within the class of streaming codes similar to Ho et al. \cite{Ho06}, the novelty of our code construction comes from the flexibility introduced with regard to the location where redundancy is placed. We show that these codes achieve capacity, yet they provide a superior throughput-delay trade-off compared to block codes by achieving significantly lower in-order delay for a given throughput.  The mathematical analysis of throughput and delay is a primary contribution of the paper, requiring the development of a number of analytic tools derived from both coding and queuing theory. By building on the work performed by \cite{tanner61}, our approach allows for a wide variety of coding scenarios that allow us to develop a novel interplay between the two areas of research. We further demonstrate the performance of these low delay codes by evaluating them using both simulations and experimental measurements.  While the mathematical analysis is confined to operation without feedback, we use the simulations and experiments to show that our proposed code construction can be easily combined with feedback to provide larger gains.

%% file: blockcodes.tex

\subsection{Streaming Codes}
We use the term \emph{streaming code} to refer to codes which do not require a bit/packet stream to be partitioned into blocks or generations before coding operations can occur.   Probably the most common examples are convolutional codes.  Most work on the performance of convolutional codes has focused on their error-correcting capabilities, particularly to bursty errors, and only a few have investigated delay. For example, Martinian et al. \cite{martinian04} and Hehn et al. \cite{Hehn09} investigate the decoding delay of convolutional codes specifically designed for channels with burst errors; and Lee et al. \cite{lee92} investigate the delay performance of orchard codes (a class of convolutional codes).  

Classical convolutional codes do not lend themselves to recoding at intermediate nodes within a network.  This has motivated work on convolutional network codes (proposed in \cite{Koetter03} and studied extensively by \cite{Erez04,Erez05,Li06,Erez10}).  Delay bounds for convolutional network codes  are provided by Guo et al. \cite{Guo13}. The delay of other codes of this type, such as those proposed by Joshi et al. \cite{joshi}, have also shown promising results. However, limitations in the models used within this work make it unclear whether or not their results can be extended to the scenarios in which we are interested. We note that T\"{o}m\"{o}sk\"{o}zi et al. \cite{toemoeskoezi_delay_2014} introduce a non-systematic code similar to our low delay code. Experimental results show significant delay gains over Reed-Solomon codes, but no analysis is provided.

\subsection{Block Codes}
\emph{Block codes} require a bit/packet stream to be partitioned into generations or blocks, each generation or block being treated independently from the rest.   
%
For example, assume that a message of size $N$ information packets $u_m$, $m=1,\ldots,N$, is partitioned into blocks or generations of size $k$ packets. Coded packets $c_{i,j}$, $i=1,\ldots,\left\lceil\nicefrac{N}{k}\right\rceil$, $j=1,2,\ldots$, are then generated separately for each block. If the code is systematic, the information packets in each block/generation are transmitted first with the coded packets transmitted after to help recover from any errors or erasures (see Figure \ref{fig:intro}(a)). If the code is not systematic, only the coded packets are transmitted. 
%
%
%
Previous work has primarily focused on the decoding delay of non-systematic constructions \cite{heidarzadeh_design_2012,lucani_broadcasting_2009,lucani_online_2010,lucani_random_2009,nistor_network_2010,sundararajan_feedback-based_2009,zeng_joint_2012,leith15}, and show that the decoding delay is essentially proportional to the block size.   The in-order delivery delay of systematic block codes is lower than that of non-systematic codes but remains essentially proportional to the block size (with a smaller pre-factor than for non-systematic codes), see Cloud et al. \cite{Cloud15} and references therein. 

\subsection{Baseline Block Code}
We use the following systematic random linear block code as a baseline for performance comparison.  This block code is similar to that considered in \cite{Cloud15} and is constructed as follows.  We generate $n-k$ coded packets, $c_{i,j}$, $i=1,\ldots,\left\lceil\nicefrac{N}{k}\right\rceil$, $j=1,\ldots,n-k$, from each block of $k$ information packets, which results in a code of rate $\nicefrac{k}{n}$. Each coded packet is a weighted random linear combination of the information packets within its block, i.e.,
\begin{align}
c_{i,j}:= \sum_{m=(i-1)k+1}^{ik} w_{i,j,m} u_{m},
\end{align}
where each information packet $u_m$ is treated as a vector in an appropriate finite field $\mathbb{F}$ of size $Q$ and the coefficients $w_{i,j,m}$ are drawn i.i.d uniformly at random from $\mathbb{F}$. It should be noted that calculations in field $\mathbb{F}$ are always carried out over symbols of size $Q$ rather than the packets themselves. 
%
A systematic code is obtained by first transmitting the $k$ information packets followed by the $n-k$ coded packets.  

Maximum likelihood decoding is used \emph{i.e.} Gaussian elimination.  Should an erasure occur, any coded packet can be used to help  reconstruct/decode the missing information packet.  For sufficiently large field size $Q$ and no more than $n-k$ erasures, any combination of $k$ packets in each block can be used to recover from the erasures with high probability.  
If more than $n-k$ erasures occurs, a decoding failure occurs and some of the information packets within the block may be unrecoverable. Otherwise, information packets will be recovered with high probability when at least $k$ packets have been successfully received.  

This code is asymptotically capacity achieving over erasure channels as block size $n\rightarrow\infty$ \cite{subra08}.  It also provides a good baseline for comparison since it is a modern, high performance code that has a number of optimality properties (\emph{i.e.,} it is representative of the best possible block code performance).  In particular, this code minimises the probability of a decoding failure for any given coding rate over a large class of block codes in addition to minimizing the decoding delay \cite{subra08}.

%% file: lowdelaycoding.tex

Our interest is in constructing a code that provides low in-order delivery delay while providing protection against errors/erasures.  We consider a systematic random linear streaming code construction that inserts coded packets at strategic locations within an uncoded/information packet stream.  

Assume that time is slotted and each slot is indexed as $t=1,2,\ldots$. Within each slot, a single packet can be transmitted. The code is constructed by interleaving information packets (i.e., uncoded packets) $u_j$, $j=1,2,\ldots$ with coded packets $c_i$, $i=1,2,\ldots$. For reasons that will be explained later, one coded packet is inserted after every $l-1$ information packets and transmitted over the network or channel. This results in a code of rate $\nicefrac{l-1}{l}$. Figure \ref{fig:scheme} illustrates this code construction.

\begin{figure}
\centering
\includegraphics[width=0.7\columnwidth]{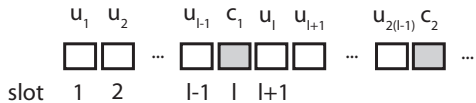}
\caption{Illustrating the setup considered.  Sequence $\{u_j\}$ of information packets is interleaved with sequence $\{c_i\}$ of coded packets (indicated as shaded) and transmitted.  Slots correspond to a single packet transmission and are indexed $1,2,\cdots$. }
\vspace{-10pt}
\label{fig:scheme}
\end{figure}

Each coded packet is generated by taking random linear combinations from a span of previously transmitted information packets $\{u_L,\ldots,u_U\}$. This span is referred to as the \textit{coding window}. We will assume in the analysis that $L=1$ and $U$ is the index of the information packet immediately proceeding the generated coded packet (i.e., $U = \left(l-1\right)i$ assuming the generated coded packet is $c_i$). Therefore, coded packet $c_i$ is generated as follows:
\begin{equation}
c_i = f_{ i}(u_1, u_2 , \dots ,u_{(l-1)i } ) := \sum_{j=1}^{(l-1)i} w_{ij} u_j,
\end{equation}
where each packet $u_j$ is treated as a vector in $\mathbb{F}_Q$ and each coefficient $w_{ij}\in \mathbb{F}_Q$ is chosen randomly from an i.i.d. uniform distribution.

Before proceeding, it should be noted that in practice the lower edge of the coding window can be determined in a variety of ways. The only constraint is that the coding window must be large enough to ensure intermediate decoding opportunities at the receiver.  For example, suppose that the receiver has received or decoded all information packets up to and including packet $u_j$. Feedback can be used to communicate this to the transmitter allowing it to adjust the coding window so that the lower edge of the window is $u_L=u_{j+1}$ for all subsequent coded packets. The generator matrix shown in Figure \ref{fig:sliding-window} illustrates this sliding window approach, where the columns indicate the information packets that need to be sent and the rows indicate the composition of the packet transmitted at any given time.  Its encoding/decoding complexity is analysed in Section \ref{sec:complexity}.


\begin{figure}
\centering
\includegraphics[width=0.8\columnwidth]{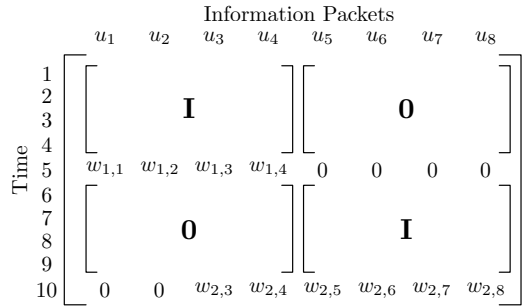}
\caption{Example generator matrix for the low delay code with sliding window showing the coefficients used to produce each packet. In this example, we assume that the transmitter has obtained knowledge from the receiver by time 10 indicating that it has successfully received/decoded packets $u_1$ and $u_2$ allowing it to adjust the lower edge of the coding window to exclude them from packet $c_2$.}
\vspace{-10pt}
\label{fig:sliding-window}
\end{figure}

The receiver decodes on-the-fly once enough packets/degrees of freedom have been received. In more detail, the receiver maintains a generator matrix $G_t$ at time $t$ which is similar to that shown in Figure \ref{fig:sliding-window} except that it is composed only of the coefficients obtained from received packets. If $G_t$ is full rank, Gaussian elimination is used to recover from any packet erasures/errors that may have occurred during transit. More details are provided in Section \ref{sec:failure-probability}. In our analysis, unless stated otherwise, we will make the standing assumption that the field size $Q$ is sufficiently large so that with probability one each coded packet helps the receiver recover from one information packet erasure.  Specifically, each coded packet row added to generator matrix $G_t$ increases the rank of $G_t$ by one. This assumption is relaxed in Section \ref{sec:failure-probability} where the probability of a decoding failure is considered; and our experimental measurements do not, of course, make this assumption.

Unlike the codes described in Section \ref{sec:blockcodes}, the code described here generates coded packets that are (i) individually streamed  between information packets (rather than being transmitted in groups of size $k$ packets) and (ii) each coded packet protects all preceding information packets (rather than just the information packets within its block). Furthermore, for a given code rate it is easy to see that this code construction should tend to decrease the overall in-order delivery delay at the receiver compared to a block code -- recall the example in Section \ref{sec:introduction}.  
%
However, the challenge is to quantify the delay performance of the low delay code. We also note that the causal construction used will limit the power of the low delay code and so its throughput performance also needs to be analysed.

%% file: inorderdelay.tex
     
\subsection{In-Order Delivery Delay}

Information packets are delivered in-order at the receiver until an erasure of an information packet occurs.   Upon erasure, in-order delivery is paused (arriving packets are buffered) until the decoder receives as many coded packets as the number of erasures, at which point in-order delivery resumes.   

\begin{figure}
\centering
\includegraphics[width=0.8\columnwidth]{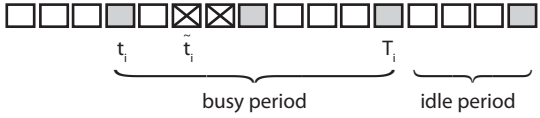}
\caption{Illustrating notation used.  Clear rectangles indicate information packets, shaded rectangles coded packets, crosses indicate erasures, coded packets are inserted every $l=4$ slots.  $t_i$ is the coded packet slot immediately preceding the information packet erasure at slot $\tilde{t}_i$ which pauses in-order delivery at the receiver, $T_1$ the coded packet slot at which in-order delivery resumes.  The information packet at slot $t_i+1$ is delivered without delay, but any information packets in slots $\{\tilde{t}_i,\cdots,T_i\}$ are delayed.  }
\vspace{-10pt}
\label{fig:notation}
\end{figure}

Let $\{\tilde{t}_i\}$ denote the sequence of slot times at which erasure of an information packet pauses in-order delivery and $\{T_i\}$ the corresponding sequence of times at which in-order delivery resumes.   Note that the $T_i$ must be a slot at which a coded packet is transmitted.  Letting $t_i=\lfloor \tilde{t}_i/l \rfloor l$ be the coded packet slot immediately preceding slot $\tilde{t}_i$, we can then define the sequence of coded packet slots $\{t_1,T_1,t_2,T_2,\cdots\}$.    See Figure \ref{fig:notation} for a schematic illustration.   Slots $\{t_i+1,t_i+2,\cdots,T_i\}$ contain information packets delayed by the $i$'th pause, plus perhaps non-delayed packets $\{t_i+1,\tilde{t}_i\}$ and this set of slots is referred to as the $i$'th ``busy'' period.  Slots $\{T_i+1,\cdots,t_{i+1}\}$ can be partitioned into intervals $\{T_i+1,T_i+l\}$,  $\{T_i+l+1,T_i+2l\}$, \emph{etc.} each of size $l$ slots and ending with a coded packet slot (since $T_i$ and $t_i$ are both coded packet slots).  Each of these intervals of $l$ slots is referred to as an ``idle'' period.     

The busy/idle period terminology is analogous with a queueing system operating in embedded time corresponding to the coded packet slots.  Information packet erasures can be thought of as queue arrivals and reception of coded packets as queue service.  Pauses in in-order delivery then correspond to periods when the queue size is non-zero.

Index the busy/idle periods by $j=1,2,\cdots$ and let $i(j)$ be the index of the pause corresponding to the $j$'th busy period (i.e., the $j$'th busy period consists of slots $\{t_{i(j)},\cdots,T_{i(j)}\}$).    With the $j$'th period we associate a random variable $S_j$, with $S_j=0$ for an idle period and $S_j=(T_{i(j)}-t_{i(j)})/l$ for a busy period (i.e., $S_j$ equals the number of coded packets transmitted before delivery resumes).  Since packet erasures are i.i.d., the busy/idle periods form a renewal process and the $\{S_j\}$ are i.i.d.  Letting $S\sim S_j$ the following theorem completely characterises the probability distribution of the busy time $S$ and is one of our main results.

\vspace{-5pt}
\begin{theorem}[Busy Time]\label{maingeniedet}
In an erasure channel with erasure probability $\epsilon$, suppose we insert a coded packet in between every $l-1$ information packets.  Assume that each coded packet can help us to recover from one erasure.  
We have:

\vspace{5pt}
\noindent $I$.  For all values of $\epsilon$ and $l$ such that $l \epsilon < 1$, the mean of the probability distribution of $S$ exists and is finite. 
%

\vspace{5pt}
\noindent $II$.
\vspace{-5pt}
\begin{equation}
p_S\left(s\right) = 
\begin{cases}
\left(1-\epsilon\right)^{l-1} & \text{for } s=0 \\
\left(l-1\right) \epsilon \left(1-\epsilon\right)^{l-1} & \text{for } s=1 \\
\frac{l-1}{s} \epsilon^{s} \left(1-\epsilon\right)^{s\left(l-1\right)} {\left(s-1\right)l \choose s-1} & \text{for } s>1 \\
0 & \text{otherwise.}
\end{cases}
\end{equation}

\noindent $III$.
\begin{align}
E\left(S\right) &= \frac{\left(l-1\right) \epsilon \left(1-\epsilon\right)^{l-1}}{1-l \epsilon} &\\
E\left(S^2\right) &= E\left(S\right) + \frac{l\left(l-1\right) \epsilon^2 \left(1- \epsilon \right)^{l} }{\left(1-l \epsilon \right)^3}
\end{align}
\end{theorem}
\begin{proof}
See Appendix.
\end{proof}

Observe that the requirement that $l \epsilon < 1$ for $S$ to have finite mean is a natural one.   The rate of the coding scheme is $R= \frac{l-1}{l}= 1- \frac{1}{l}$.  Since this rate of transmission should be less than the channel capacity, we require $R< 1- \epsilon$, and so $l \epsilon < 1$.  

We also emphasise that the in-order delivery delay expressions in Theorem \ref{maingeniedet} are exact (they are not bounds) and have an easy to evaluate closed-form (they are not combinatorial in nature).   This is notably different from previous analysis of in-order delivery delay and derives from the favourable structure of the low-delay coding scheme.

To continue the analysis, we introduce the random variable $S^+= \min \{S , 1\}. $ $S^+$ helps us to count the number of intervals in the communication interval $T$.  It is straightforward to compute the probability distribution of $S^+$ as follows:
 
\begin{corollary}\label{maingeniecorr}
Let $S^+ = \min \{S, 1\}.$ We have:
   
\vspace{5pt}
\noindent $I$. For all values of $\epsilon$ and $l$ such that $l \epsilon < 1$, the mean of the probability distribution of $S^+$ exists and is finite. 

\vspace{5pt}
\noindent $II$.
\begin{equation}
p_{S^+}\left(s\right) =
\begin{cases}
\left(l \epsilon +1-\epsilon \right) \left(1-\epsilon \right)^{l-1} & \text{for } s=1 \\
\frac{l-1}{s} \epsilon^{s} \left(1-\epsilon\right)^{s\left(l-1\right)} {\left(s-1\right)l \choose s-1} & \text{for } s>1 \\
0 & \text{otherwise.}
\end{cases}
\end{equation}

\noindent $III$.
\begin{equation}
E\left(S^+ \right) = \frac{ \left(1-\epsilon \right)^{l}}{1-l \epsilon}.
\end{equation}
\end{corollary}
         
Combining Theorem \ref{maingeniedet} and Corollary \ref{maingeniecorr}  with the following result allows us to obtain a simple closed-form bound on the mean in-order delivery delay:

\begin{figure}
\centering
\includegraphics[width=0.5\columnwidth]{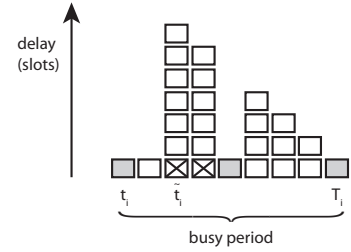}
\caption{Illustrating delay introduced by erasures.   In this example the information packet at slot $\tilde{t}_i$ is delayed by 6 slots, the information packet at slot $\tilde{t}_i+1$ by 5 slots, the information packet at $\tilde{t}_i+3$ by 3 slots and so on.   It can be seen that the sum-delay is the area under a triangle of base $T_i-t_i$ slots and height $T_i-t_i$ slots, less the area associated with any coded packets.}
\vspace{-10pt}
\label{fig:meandelayproof}
\end{figure}
\begin{theorem}[In-Order Delivery Delay]\label{maindelay}
At the receiver, the asymptotic mean in-order delivery delay for information packets is upper bounded by $\frac{E[S^2](l-1)}{2E[S^+]} $ slots.
\end{theorem}
\begin{proof}
See Appendix.
\end{proof}
The comparison of this upper bound for in-order delay with simulated results is provided in Fig. \ref{fig:bound}. It can be seen that the upper bound is tight at both low and high coding rates. Furthermore, it is reasonably tight at intermediate coding rates. Despite its simple form, it is therefore quite powerful.

\begin{figure}
\centering
\includegraphics[width=0.9\columnwidth]{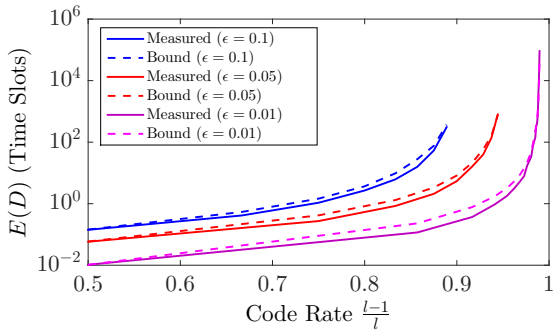}
\caption{Measured mean in-order delivery delay and upper bound versus code rate $\nicefrac{l-1}{l}$ for different i.i.d. packet erasure rates $\epsilon$.}
\vspace{-10pt}
\label{fig:bound}
\end{figure}

\vspace{-5pt}
\subsection{In-Order Delivery Delay where Coded Packets are Placed in Groups}

Before proceeding with determining the throughput of the low delay code, we first provide justification of why we chose to only send one code packet at a time to recover from losses. Consider the coding scheme where $c$ coded packets are transmitted after every $l_g-c$ information packets (\emph{i.e.} coded packets are transmitted in groups of $c$ packets).  Note that this reduces to the low-delay coding scheme when $c=1$.  


Following the same steps as before, we define the process $S_g$ as the busy period of the decoding operation.  While extending the proof technique used previously is difficult in this case, results from lattice theory can be used to prove the following:


\begin{theorem}[Busy Time for Group Coded Packets]\label{maingeniedetgroup}
In an erasure channel with erasure probability $\epsilon$, suppose we insert $c$ coded packets in between every $l_g-c$ information packets.  Assume that each coded packet can help us to recover from $c$ erasures.  
We have:

\vspace{5pt}
\noindent $I$.  For all values of $\epsilon$ and $l_g$ such that $l_g \epsilon < c$, the mean of the probability distribution of $S_g$ exists and is finite. 

\vspace{5pt}
\noindent $II$.
\begin{multline}
p_{S_g}\left(s\right) = \\
\begin{dcases}
\left(1-\epsilon \right)^{l_g-c} & \text{for } s=0 \\
\sum_{i=1}^c \sum_{j=0}^{c-i} {l_g - c \choose i} {c \choose j} f\left(\epsilon,i+j,l_g\right) & \text{for } s=1 \\
\sum_{i=0}^{c-1} NP\left(l_g,c,s,i\right) f\left(\epsilon,sc-i,sl_g\right) & \text{for } s>1 \\
0 & \text{otherwise.}
\end{dcases}
\end{multline}
where $f\left(\epsilon,x,y\right) = \epsilon^{x}\left(1-\epsilon\right)^{y-x}$ and the exact value of $NP(l_g,c,s,i)$ is computed in the proof of the theorem. 
\end{theorem}
\begin{proof}
See Appendix.
\end{proof}

Theorem \ref{maingeniedetgroup} does not provide a closed-form expression of the mean in-order delay when coded packets are sent in groups, but it does allow numerical calculation of the in-order delivery delay. 

Fig. \ref{fig:groupcode} shows the calculated delay per information packet vs $c$.   In order to provide a fair comparison for different choices of $c$ the coding rate is held constant, \emph{i.e.} $1-\nicefrac{c}{l_g}=1-\nicefrac{1}{l}$ is held constant as $c$ is varied.  It can be seen that the delay is an increasing function of $c$. In other words, the delay is minimized when only a single coded packet is transmitted at a time (i.e., $c=1$). 
%
%
\begin{figure}
\centering
\includegraphics[width=0.9\columnwidth]{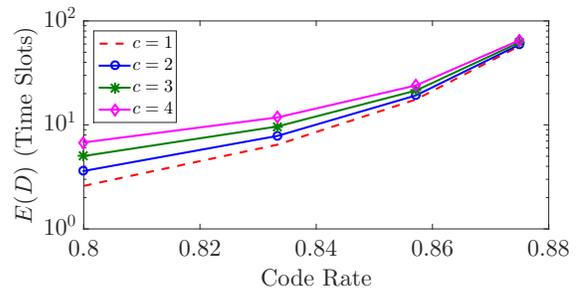}
\caption{Average in order delivery delay per information packet, $c$ is the number of coded packets at the end of each interval, $\epsilon = 0.1$.}
\vspace{-10pt}
\label{fig:groupcode}
\end{figure}

\vspace{-5pt}
\subsection{Throughput and Rate}
We can expect that the improved delay performance of the low delay coding scheme carries a throughput price.   However, it turns out that this price is a small one.  For a stream of $N$ packets, a decoding error may occur since it is possible that a burst of errors near the end of the stream may not allow for sufficient time to transmit the necessary coded packets to recover the lost information packets. That is, a number of information packets at the end of a transmission may be lost.   Define the good throughput $GT$ as the ratio of the number of information packets delivered to the receiver and the number of packets transmitted by the transmitter. The good throughput is a random variable and its behavior is characterized in the following theorem:
 \begin{theorem}[Throughput]\label{throughputthm}
 Consider the transmission of a coded stream of length $N$ over an erasure channel and that $l\epsilon<1$.  
 For any $R_0 < 1- \frac{1}{l}$ and $\delta >0 $, there exist an $N$ large enough such that 
 \begin{equation}
 \Pr{(GT>R_0)} > 1- \delta.
 \end{equation}
 \end{theorem}
 \begin{proof}
See Appendix.
\end{proof}

Recall that the rate of transmission of the low delay coding scheme is $\frac{l-1}{l} = 1 - \frac{1}{l}$ and that the capacity of the erasure channel is $1-\epsilon$, so the condition $l\epsilon<1$ allows all coding rates up to the channel capacity.   Theorem \ref{throughputthm} therefore tells us that the low delay code is asymptotically capacity achieving as $N\rightarrow \infty$.  

In other words, the fraction of information packets lost can be made arbitrarily small for sufficiently long transfers.   This is because  the length of the possible failure at the end of a transmission is independent of the length of the transmission. Therefore, the sacrifice in the good throughput becomes negligible as the length of the transmission grows. Of course, losing any information packet is undesirable. However, an easy extension of the coding scheme is to send a small number of additional coded packets at the end of a transfer (or alternatively to use feedback, see below).  This will allow any straggling information packet erasures to be recovered without adversely affecting throughput (or mean delay) when the connection size $N$ is large.

%% file: failure.tex

\subsection{Computation of Decoding Failure Probability}
Recall that following an erasure at time $\tilde{t}_i$ the receiver pauses in-order delivery of packets. As packets are received, the receiver maintains a generator matrix, with $G_t$ denoting this generator matrix at time $t>\tilde{t}_i$.  This matrix contains $N_t \le t-\tilde{t}_i$ rows where $N_t$ is a random variable with value equal to the number of packets received in the interval from time $\tilde{t}_i$ to time $t$.  Reception of information packet $j$ adds a row to $G_t$ with element $j$ equal to 1 and all other elements equal to zero.  Reception of coded packet $i$ adds a row with elements 1 through $(l-1)i$ equal to coefficients $w_{ij}$ and other elements equal to zero.  When $G_t$ reaches full rank, decoding succeeds and all information packets sent up to time $t$ are recovered. An example of $G_t$ is shown in Figure \ref{fig:sliding-window}.

Until now, we have assumed that every received packet increases the rank of $G_t$. In fact there is a small probability that this does not occur since the coding coefficients are selected uniformly at random. For example, should the coefficients of two packets be the same, it is the same as if the received coded packet had been erased resulting in the appearance of a larger erasure rate. While obtaining an analytic expression for this increase in erasure rate is difficult, it can be readily calculated numerically. Table \ref{fig:decfailmod} shows calculated values of the decoding failure probability for a range of field sizes $Q$ and values of code rate with parameter $l$ on a channel with erasure rate $\epsilon=0.1$. The table shows that the decoding failure probability is extremely small, even for small field sizes such as $Q=2$ (i.e., the binary field). As a result, decoding failures can be neglected for most practical purposes and the simplifying assumption in the preceding section is reasonable.

\setlength\extrarowheight{3pt}
\begin{table}[htdp]
\begin{center}
\begin{tabular}{cc|c|c|c|c|c|}
\cline{3-7}
& & \multicolumn{5}{ c| }{$Q=2^{x}$} \\ \cline{3-7} 
& $x$ & 2& 3 & 4 & 5 & 6 \\ \cline{1-7}
\multicolumn{1}{ |c  }{\multirow{3}{*}{$l$} } & \multicolumn{1}{ |c| }{$5$} &  $10^{-16.71}$ &   $10^{-20.2}$ &   $\sim 0$  & $\sim 0$   &  $\sim 0$ \\ \cline{2-7}
\multicolumn{1}{ |c  }{} & \multicolumn{1}{ |c| }{$6$} &  $10^{-14.32}$ &   $10^{-16.22}$ &   $10^{-18.68}$ &  $10^{-22.74}$ &  $\sim 0$  \\ \cline{2-7}
\multicolumn{1}{ |c  }{} & \multicolumn{1}{ |c| }{$7$} & $10^{-12.81}$ &   $10^{-13.31}$ &   $10^{-16.21}$ &  $10^{-19.03}$ &  $\sim 0$\\ \cline{1-7}
\end{tabular}
\end{center}
\caption{Decoding failure probability per packet, $\epsilon =0.1$ and $\sim 0$ indicates zero to within numerical precision.}
\vspace{-10pt}
\label{fig:decfailmod}
\end{table}%

\vspace{-5pt}
\subsection{Analytic Bound}
While the numerical calculation is significantly less conservative, we note that it is also possible to upper bound the decoding failure probability. Suppose that $k$ erasures occur where the erasure pattern is admissible (as defined in Lemma \ref{lemmaTanner}{\footnote {Reminder:  The term \emph{admissible} is used here to mean that at least  two erasures should happen during a time $t=l$, at least three within $2l$, and so on, so that the decoding process will remain activated for the whole of the time $t= kl$.}}), and decoding is attempted after receiving exactly $k$ coded packets. Let $E_1, E_2, \dots , E_k$ denote the number of erasures in each $l$-interval, and $G'\in\mathbb{R}^{k \times k}$ be the admissible decoding matrix obtained by removing the rows and columns in generator matrix $G_{kl}$ that are associated with received information packets so that the dimension of the decoding matrix is equal to the number of erasures.  The number of non-zero elements in each row $i$ is equal to $\sum_{z=1}^i E_z$. Furthermore, the number of erasures in any $k'$ $l$-interval for any $k'$ less than $k$ is strictly greater than $k'$ since the decoding does not stop before coded packet $k$.



 \begin{theorem}[Decoding failure for $S=k$]\label{theoremfailure}
 Consider an admissible decoding matrix $G'$, and assume that its elements are drawn identically and independently, uniformly at random from a field with size $Q$. The probability that the matrix is full rank is bounded as follows:
 \begin{align}
 \Pr ( \textrm{rank}(G') = k ) & \leq \prod_{j=0}^{k-1}\left(1-\frac{1}{Q^{k-j}}\right) \\
  \Pr ( \textrm{rank}(G') = k ) & \geq \frac{Q}{Q+1} \left(1-\frac{1}{Q^2}\right)^{k}. 
 \end{align}
 \end{theorem}
\begin{proof}
See Appendix.
\end{proof}

Note that these upper and lower bounds coincide for the cases where $k=1$ and $k=2$ since there exists only one admissible decoding matrix in each case.

\begin{theorem}[Decoding failure for a stream of length $N$]\label{theoremfailstr}
The decoding failure(DF) probability in a stream of length $N$ satisfies
\begin{multline}
\lim_{N_t \to \infty} \frac{1}{N_t} Pr \left(\textrm{DF}\right)  \leq \\ \frac{\left(1-l \epsilon\right)}{l\left(1-\epsilon\right)^{l}}  
\sum_{i=1}^{\infty} \left(1-\frac{Q}{Q+1} \left(1-  \frac{1}{Q^{2}} \right)^{i}\right) p_S\left(i\right).
\end{multline}
\end{theorem}
\begin{proof}
See Appendix.
\end{proof}

\begin{corollary}[Decoding Failure Probability] \label{corollaryfailure}
The decoding failure probability satisfies
\begin{multline}
\lim_{N_t \to \infty} \frac{1}{N_t} Pr (\textrm{DF})  \leq  \frac{\left(1-l \epsilon\right)}{l\left(1-\epsilon\right)^{l}} \cdot \\
 \quad  \cdot \left( \frac{Q}{Q+1} \left(1-\epsilon_0\right)^{l-1} - \left(1-\epsilon\right)^{l-1} + \frac{1}{Q+1} \right), 
 \end{multline}
where $\epsilon_0$ is the solution to the equation
\begin{equation}
\epsilon_0 \left(1- \epsilon_0\right)^{l-1}=\left(1-\frac{1}{Q^2}\right) \epsilon \left(1- \epsilon\right)^{l-1}.
\end{equation}
Since the function $f\left(\epsilon\right) = \epsilon \left(1- \epsilon\right)^{l-1}$ is increasing for $\epsilon < \nicefrac{1}{l}$, the solution always exists and $\epsilon_0 < \epsilon$. 
\end{corollary}

As $Q$ becomes larger, the decoding failure probability goes to zero. Table \ref{tab:dfbound} gives values for the upper bound in Corollary \ref{corollaryfailure} for a channel with erasure probability $\epsilon=0.1$.  Comparing with Table \ref{fig:decfailmod} it can be seen that the bound is not tight.

\begin{table}[htdp]
\caption{The upper bound on decoding failure probability per packet, $\epsilon =0.1$}
\begin{center}
\begin{tabular}{cc|c|c|c|c|c|}
\cline{3-7}
& & \multicolumn{5}{ c| }{$Q=2^{8x}$} \\ \cline{3-7} 
& $x$ & $1$ & $2$ & $4$ & $10$ & $20$ \\ \cline{1-7}
\multicolumn{1}{ |c  }{\multirow{3}{*}{$l$} } & \multicolumn{1}{ |c| }{$5$} &   $10^{-2.24}$ &   $10^{-4.65}$ &   $10^{-9.46}$ &  $10^{-23.45}$ &  $10^{-47.53}$ 
 \\ \cline{2-7}
\multicolumn{1}{ |c  }{} & \multicolumn{1}{ |c| }{$6$} &  $10^{-2.13}$ &   $10^{-4.54}$ &   $10^{-9.36}$ &  $10^{-23.42}$ &  $10^{-47.51}$  \\ \cline{2-7}
\multicolumn{1}{ |c  }{} & \multicolumn{1}{ |c| }{$7$} & $10^{-2.09}$ &   $10^{-4.50}$ &   $10^{-9.31}$ &  $10^{-23.43}$ &  $10^{-47.52}$   \\ \cline{1-7}
\end{tabular}
\end{center}
\vspace{-15pt}
\label{tab:dfbound}
\end{table}%



%% file: complexity.tex
The encoding and decoding complexity is dependent on the management of the coding window. As already noted in Section \ref{sec:lowdelaycoding}, to limit the complexity we can use a sliding window approach that keeps track of the decoding process at the receiver. This results in a complexity that is polynomial in $E(S)$ and is considered in more detail below. Alternatively, the encoding and decoding might be constrained to the $k l$ preceding information packets. This will result in a small probability of decoding failure $\Pr(S>k)$, but this will go to zero exponentially as $k$ grows (see Theorem \ref{maingeniedet}).  

The complexity of the sliding window scheme depends on the process $S$. Assuming that $S=k$, Gaussian elimination can be performed at each step requiring approximately $\frac{2 k^3}{3}$ arithmetic operations.
 Using the elementary renewal theorem, we have:

\begin{theorem}
The total number of arithmetic operation $C_d$ in a stream of length $N_t$ which consists of $N_i = \nicefrac{(l-1)N_t}{l}$ information packets, satisfies the following:
\begin{equation}
\lim_{N_i \to \infty} \frac{1}{N_i} C_d = \frac{3}{2} \frac{1- l\epsilon}{\left(l-1\right)\left(1-\epsilon\right)^l} E\left(S^3\right),
\end{equation}
where 
\begin{multline}
E\left(S^3\right) = \frac{l\left(l-1\right)\epsilon^2 \left(1-\epsilon\right)^l \left(2 - 2 \epsilon - 2 l \epsilon^2 + l \epsilon + l^2 \epsilon^3\right)}{\left(1-l\epsilon \right)^5} \\
+ E\left(S^2\right).
\end{multline}
\end{theorem}
\begin{proof}
The proof is similar to the proof of the theorem \ref{theoremfailstr} where the computation of $E(S^3)$ is similar to the computation of $E(S^2)$ in Theorem \ref{maingeniedet}. 
\end{proof}

Note that the number of arithmetic operations per information packet becomes large as the code rate approaches the capacity. However, the complexity in the low delay regime (i.e., the regime where the low delay code should be used) results in a complexity that is manageable (see Table \ref{tab:complexity}). 


\begin{table}[htdp]
\caption{The average number of arithmetic operations performed in the decoder per information packet}
\begin{center}
\begin{tabular}{cc|c|c|c|c|c|}
\cline{3-7}
& & \multicolumn{5}{ c| }{$l=\nicefrac{x}{\epsilon}$} \\ \cline{3-7} 
& $x$ & $0.5$ & $0.6$ & $0.7$ & $0.8$ & $0.9$ \\ \cline{1-7}
\multicolumn{1}{ |c  }{\multirow{2}{*}{$\epsilon$} } & \multicolumn{1}{ |c| }{$0.02$} &   $0.67$ &   $1.93$ &   $7.11$ &  $41.74$ &  $769.58$ 
 \\ \cline{2-7}
\multicolumn{1}{ |c  }{} & \multicolumn{1}{ |c| }{$0.1$} &  $3.13$ & $8.87$ &   $32.56$ &  $190.96$ &  $3525$  \\ \cline{1-7}
\end{tabular}
\end{center}
\vspace{-15pt}
\label{tab:complexity}
\end{table}


%% file: sim_exp_results.tex
This section uses stochastic simulations and experimental results to supplement the analysis provided in previous sections.  While the low delay code construction falls within the category of streaming codes, we compare its performance with systematic block codes because of their widespread use and so that the block code can be used as a baseline for comparison with other streaming codes. Many of which also use block codes for purposes of comparison. We consider both the case when feedback is not used to quickly recover from packet losses (i.e., open-loop) and the case when it is used (i.e., closed-loop).


\vspace{-5pt}
\subsection{Simulation Results}

The block and low delay codes are generated as described in Sections \ref{sec:blockcodes} and \ref{sec:lowdelaycoding} respectively. We assume that the alphabet size $Q$ for both codes is large enough so that the probability of receiving linearly dependent packets is zero. Therefore, a decoding error only occurs when the channel/network erases more than $n-k$ packets when a block code is used where $k$ is the block size and $c=\nicefrac{k}{n}$ is the code rate. 

Simulations are carried out for i.i.d. erasures with packet erasure probability $\epsilon$; and for correlated packet erasures described by a two state Markov chain. The Markov chain has a ``good'' state (i.e., state $G$) with packet erasure rate $\epsilon=0$, a ``bad" state (i.e., state $B$) with packet erasure rate $\epsilon=1$, and the transition probability matrix
\begin{equation}
P_{GC}=\begin{bmatrix}
			1-\gamma & \gamma \\
			\beta & 1-\beta
		\end{bmatrix}.
\end{equation}
Two parameters are used to generate the transition probabilities $\beta$ and $\gamma$: the steady-state probability of state $B$, $\pi_B=\nicefrac{\gamma}{\gamma+\beta}$; and the expected burst length, $E\left(L\right)=\nicefrac{1}{\beta}$. Within the figures presented in this section, we will refer to the i.i.d. model by referencing either erasure rate $\epsilon$ or the 2-tuple $(\pi_B, E(L)=1)$ (note that $\pi_B$ equals the erasure rate). When the correlated loss model is assumed, the 2-tuple $(\pi_B,E(L)>1)$ will be used.

\begin{figure}
\centering
\subfigure[I.I.D. Packet Losses with $\epsilon=0.05$]{
\includegraphics[width=0.95\columnwidth]{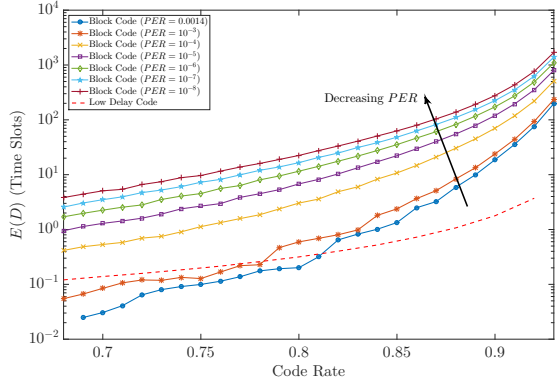}}
\subfigure[Correlated Packet Losses with $\pi_B=0.1$]{
\includegraphics[width=0.95\columnwidth]{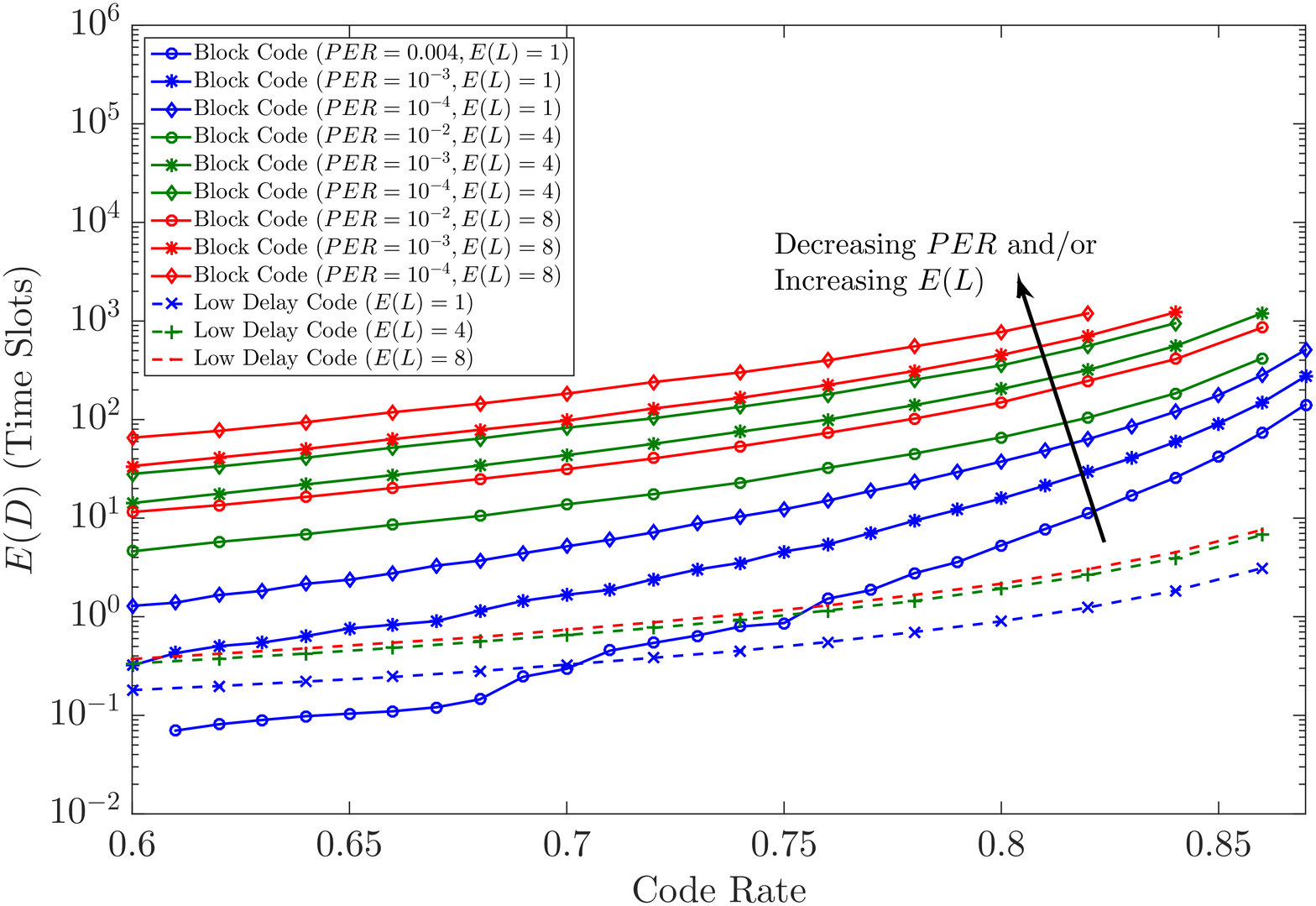}}
\caption{Open-loop mean in-order delivery delay, $E(D)$, versus code rate for a systematic block code and the low delay code ($E(L)$ is the expected packet erasure burst duration).}
\vspace{-15pt}
\label{fig:vsopenloop}
\end{figure}

We first consider the open-loop case where feedback communicating the successful reception of a packet is unavailable. In this case there is always non-zero probability of decoding failure for the block code, corresponding to the event that the number of packet losses over a block exceeds the number of codes packets in the block.  Therefore, there remains a non-zero packet erasure rate ($PER$) after coding is applied.  For the low-delay block code this probability is close to zero for sufficiently large connections (see Section \ref{sec:failure-probability}). Specifically, the $PER$ of the low delay code is essentially zero.

A comparison of the two codes is shown in Figure \ref{fig:vsopenloop}(a) where the mean in-order delivery delay $E(D)$ is plotted as a function of the coding rate for $\pi_B=0.05$ and $\pi_B=0.1$ when $E(L)=1$ (i.e., the packet erasures are i.i.d.).  A similar comparison is shown in Figure \ref{fig:vsopenloop}(b) for correlated packets losses where $E(L)\geq 1$.  Each solid line shows the delay of the block code vs the coding rate when the block size is adjusted to hold the packet erasure rate constant (as the coding rate increases the block size must also grow to achieve the same $PER$). The mean in-order delivery delay for the low delay code is shown as a dotted line.



\begin{figure}
\centering
\subfigure[$\epsilon=0.05$]{
\includegraphics[width=0.95\columnwidth]{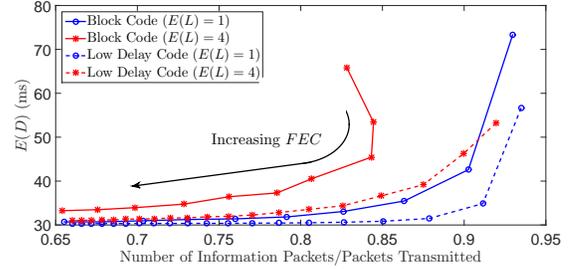}}
\subfigure[$\epsilon=0.1$]{
\includegraphics[width=0.95\columnwidth]{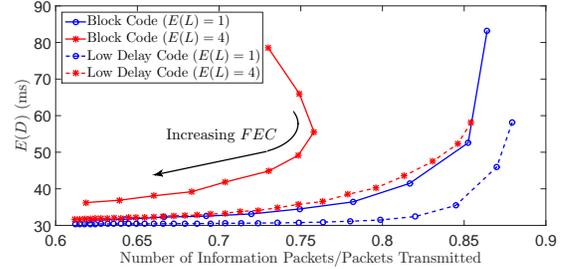}}
\caption{Mean in-order delivery delay, $E(D)$, versus code rate on a $25$ $Mbps$ link with an $RTT$ of $60$ $ms$. Both a systematic block code and the low delay code are shown where feedback is used to signal retransmissions and packet erasures are correlated ($E(L)$ is the expected packet erasure burst duration). The non-uniqueness in the abscissa for the block code when $E(L)=4$ occurs when the probability of decoding each generation or block without needing to retransmit additional degrees of freedom is very low.   
}
\vspace{-15pt}
\label{fig:vsclosedloop_corr}
\end{figure}

It can be seen that the low delay code achieves a smaller in-order delivery delay than block codes for the cases that are of the most interest.  The reduction in delay is substantial, being on the order of a magnitude or more for any given code rate.  The regime where this is not the case is when the rate of decoding failures for the block codes is large (e.g. $PER \geq 10^{-3}$ when compared to a channel packet erasure rate of $0.05$) and the coding rate is small (e.g., $c\leq 0.8$).  Recall that the low-delay code has a $PER\approx 0$, so the delay comparison is not really fair in within this regime.  Also note that the block code in this regime typically has a  block size of only 1 to 3 packets, which is far smaller than is usual for block codes.  As both the code rate and block size are increased, quantization due to these small block sizes results in the fluctuations shown in the delay-rate curves within Figure \ref{fig:vsopenloop}.

It can also be seen from Figure \ref{fig:vsopenloop}(b) that the block code's in-order delivery delay is much more sensitive to correlated losses than the low delay code's delay. This is a result of the low-delay code removing the requirement to partition the packet stream into blocks or generations prior to coding.

Simulation results for the closed-loop case are shown in Figure \ref{fig:vsclosedloop_corr}. The major difference between this and Figure \ref{fig:vsopenloop} is that feedback is used to help communicate the receiver's need for additional degrees of freedom. When considering the block code, feedback is used to initiate retransmissions in the form of coded packets if a block cannot be decoded. These retransmissions occur until every block can be decoded and delivered.  When considering the low delay code, feedback is used to adjust the code rate to ensure frequent decoding opportunities. 

While the gain in delay is not as pronounced as the open-loop case,
the low delay code achieves a lower in-order delivery delay than the block code over the entire range of code rates (measured as the total number of information packets divided by the total number of both transmitted information and coded packets). Furthermore, the figure highlights the inability of block codes to recover from correlated losses. This is shown by the non-uniqueness in the abscissa. Each curve is generated by increasing the forward error correction (FEC) code rate. For smaller FEC rates where a decoding error occurs, retransmissions are necessary resulting in larger delays. As the FEC is increased, the probability of requiring retransmissions to decode a block decreases resulting in lower delay.  




\vspace{-5pt}
\subsection{Experimental Results}

We implemented both the low delay code and the systematic random linear block code within the coded TCP transport protocol, CTCP \cite{kim}.  Use of error-correction coding at the transport layer to reduce delay was one of the original motivations for the present work and is currently the subject of much interest, (e.g., as part of the Google QUIC protocol). Since delayed ACK feedback is available at the transport layer, the setup is similar to the closed-loop case above. The congestion control used in CTCP was disabled by fixing the congestion window ($cwnd$) to be equal to the path bandwidth-delay product (BDP) in each of the experiments in order to focus on the coding performance.



\begin{figure}
\centering
\includegraphics[width=0.7\columnwidth]{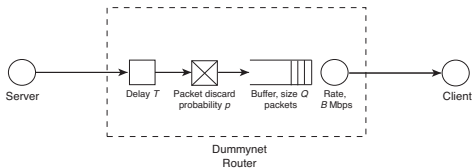}
\caption{Testbed setup.}
\vspace{-15pt}
\label{fig:testbed}
\end{figure}

The testbed used consists of three commodity servers (Dell Poweredge 850, 3GHz Xeon, Intel 82571EB Gigabit NIC) connected via a router and gigabit switches (Figure \ref{fig:testbed}).  Both the server and client machines ran a Linux 2.6.32.27 kernel, while the router ran a FreeBSD 4.11 kernel. \texttt{ipfw}-\texttt{dummynet} was used on the router to configure various propagation delays $T$, packet loss rates $p$,  queue sizes $Q$ and link rates $B$.  As indicated in Figure \ref{fig:testbed}, packet losses in \texttt{dummynet} occur before the rate constraint, not after, so that the bottleneck link capacity $B$ is not reduced.   Finally, data traffic is generated using HTTP traffic with \texttt{apache2} (version 2.2.8) and \texttt{wget} (version 1.10.2).


Figure \ref{fig:4} compares the mean in-order packet delivery delay, $E(D)$, for the low delay code and the systematic block code with different block sizes using the same coding rate.   Observe that there is an ``optimal'' block size where the block code achieves the lowest delay (a behavior also highlighted in \cite{Cloud15}), and the low delay code achieves a mean delay that is about half of of this value for the link used in the experiment. Assuming that the block size is not tuned with respect to the path characteristics, the delay improvement can be significantly larger. Furthermore, the two time history snapshots shown in Figure \ref{fig:5} also verify that the delay due to head-of-line blocking is reduced.

\begin{figure}
\centering
\includegraphics[width=1\columnwidth]{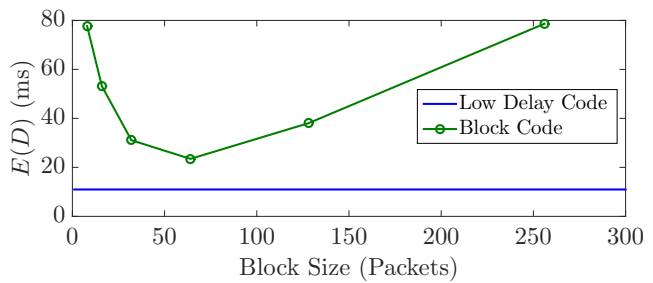}
\caption{Mean in-order packet delivery delay, $E(D)$, versus the block size for a random linear block code and low delay code.  A link rate of $25$ $Mbps$, $RTT$ of $60$ $ms$, packet erasure rate of $10\%$, redundancy of $15\%$, and a $cwnd$ fixed at the $BDP$ (125 packets) was used.}
\vspace{-10pt}
\label{fig:4}
\end{figure}


\begin{figure}
\centering
\subfigure[Block Code ($k = 64$)]{\includegraphics[width=0.45\columnwidth]{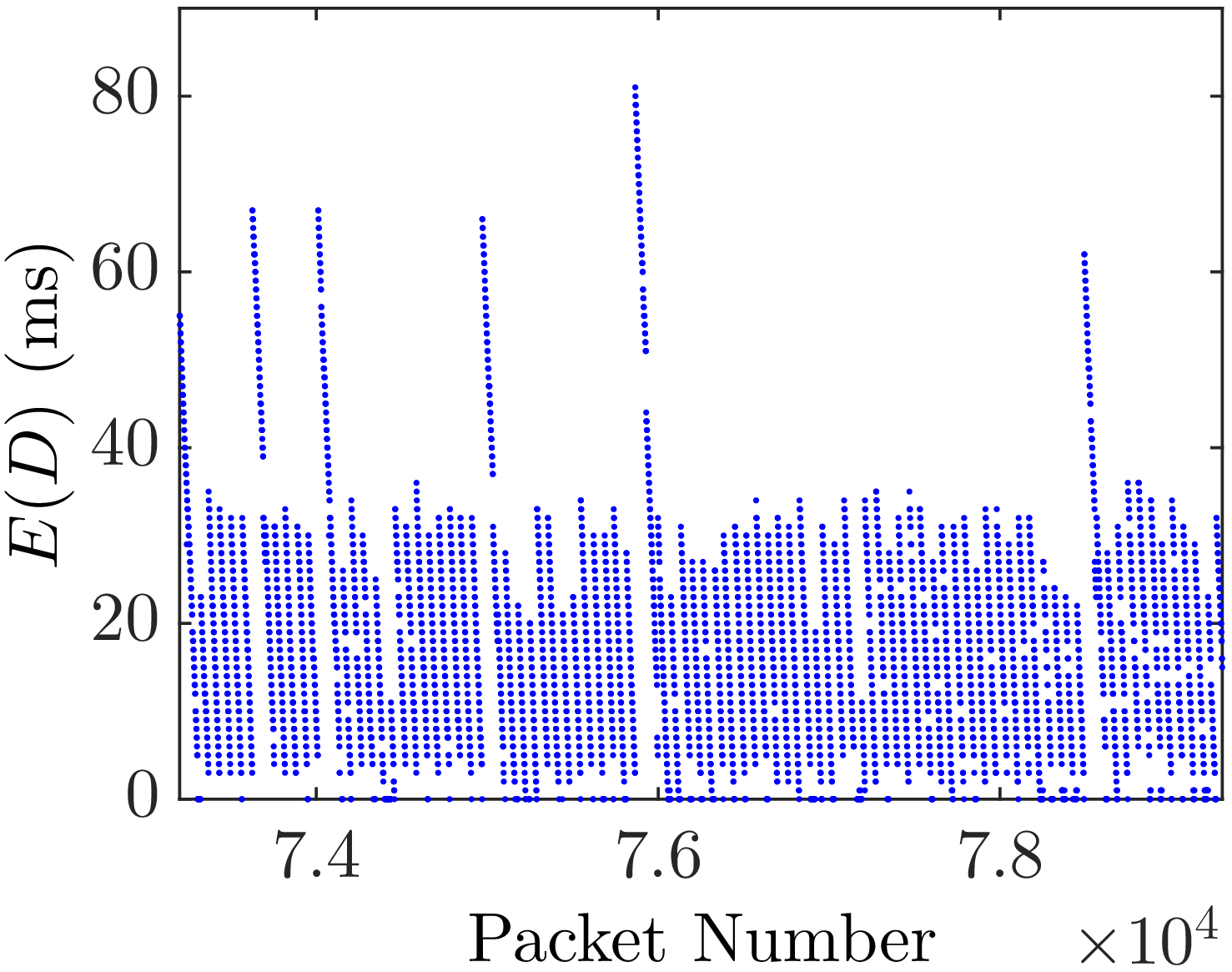}}
\subfigure[Low Delay Code]{\includegraphics[width=0.45\columnwidth]{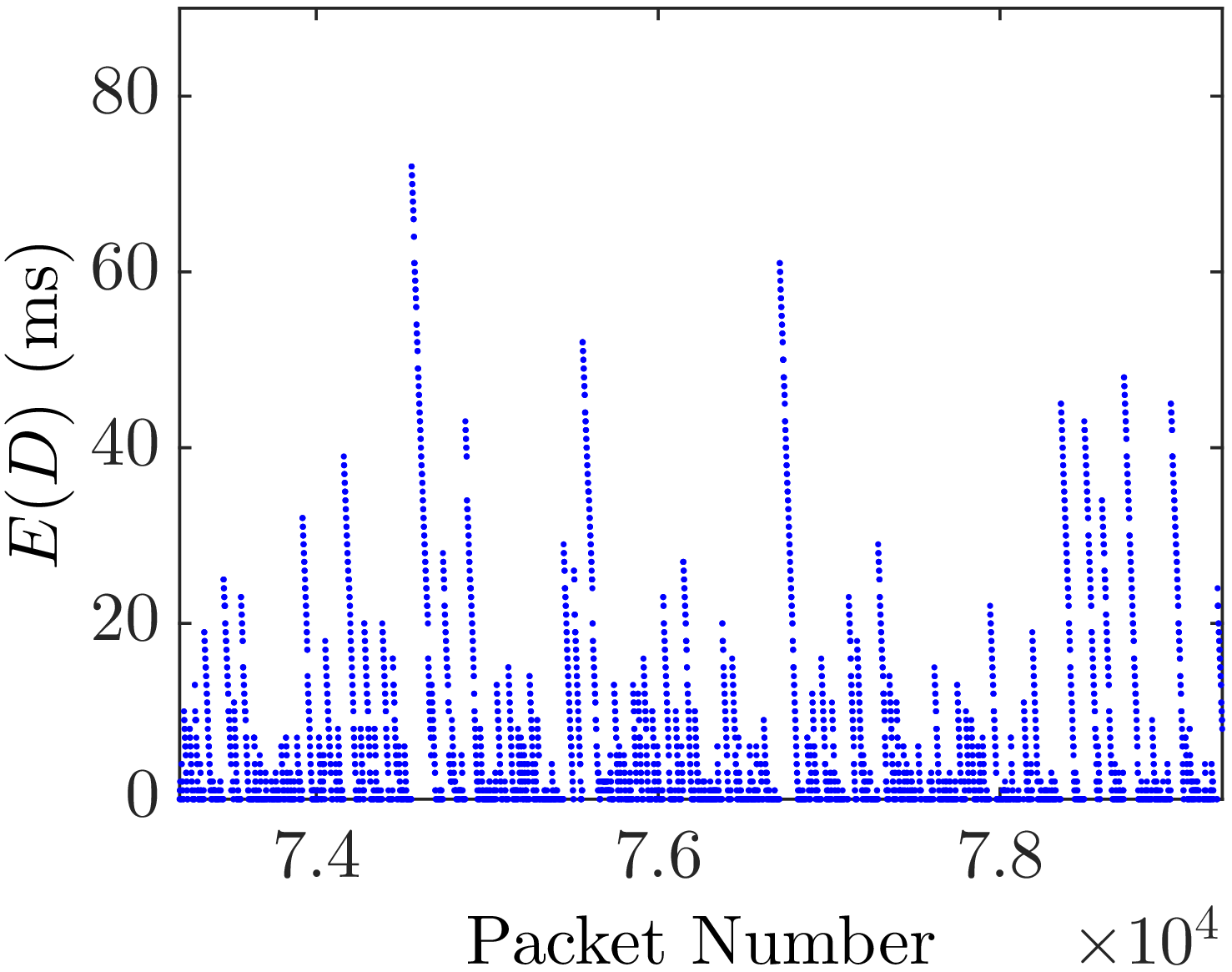}}
\caption{Time history snapshots of the in-order delivery delay, $E(D)$, for a random linear code with $k=64$ packets and the low delay code.  A link rate of $25$ $Mbps$, $RTT$ of $60$ $ms$, packet erasure rate of $10\%$, redundancy of $15\%$, and a $cwnd$ fixed at the $BDP$ was used.}
\vspace{-15pt}
\label{fig:5}
\end{figure}


Finally, the mean in-order delivery delay versus the reciprocal of the code rate is shown in Figure \ref{fig:7}. Measurements are shown for the block code over a range of block sizes, in addition to measurements taken using the low delay code. The initial code rate for both the block and low delay code was set to be $\nicefrac{1}{1-\epsilon}$.  Since feedback is used, there are no decoding failures for either type of code. As expected based on \cite{Cloud15} and Figure \ref{fig:4}, the block code's in-order delivery delay is a function of the coding rate and block size. For the initial coding rate used in the experiment, the minimum delay occurs for block sizes between $k=32$ and $k=64$. 

Within our experiments, the low delay code achieves a smaller delay over the entire range of coding rates.   Specifically, the in-order delay is always smaller when using the low-delay code for a given throughput.  Recall that the block code used for comparison is a modern, high performance code, and we expect similar behaviours when other types of block codes are used.   While not shown here, similar results were obtained for a wide range of network bandwidths and RTTs.   



\begin{figure}
\centering
\includegraphics[width=0.95\columnwidth]{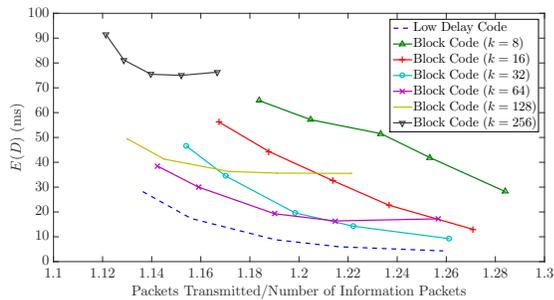}
\caption{Mean in-order packet delivery delay, $E(D)$, versus the number of packets transmitted divided by the number of information packets for a systematic random linear block code and low delay code. A link rate of $25$ $Mbps$, $RTT$ of $60$ $ms$, loss rate of $10\%$, and a $cwnd$ fixed at the $BDP$ was used. }
\vspace{-15pt}
\label{fig:7}
\end{figure}

%% file: conclusion.tex
We introduced a class of streaming codes for packet erasure channels that is capacity achieving and provides a superior throughput-delay trade-off compared to block codes. By introducing the flexibility of where and when to place redundancy, our code construction achieves significantly lower in-order delay for a given throughput.  Furthermore, the mathematical analysis of throughput and delay is a primary contribution of the paper. This analysis required the development of a number of novel analytic tools based on both queuing and coding theory.  Finally, simulation and experimental results were provided showing the performance of our code construction. While the mathematical analysis is confined to operation without feedback, we showed that feedback can naturally be combined with low delay codes in both the simulation and experimental results.

While this paper focused on many of the analytical aspects of low delay codes, future work is still required. The decoding failure probability discussed in Section \ref{sec:failure-probability} can be further reduced by treating our code like a standard rateless code when terminating a packet stream. The interaction between congestion control and the low delay code, as well as its performance when TCP is used, is also of interest. Finally, we believe, based on work done by \cite{Teerapittayanon12}, that this code construction can help improve network performance by removing or replacing complex lower layer reliability schemes such as hybrid ARQ (H-ARQ).



%% file: app1.tex
The following lemma corresponds to a result of Tanner \cite{tanner61} for the busy time distribution of a G/D/1 queue. 

\begin{lemma}[Tanner61] \label{lemmaTanner}
Consider a decoding procedure at time $t=0$ with $r$ erasures already occurred in $t<0$. Suppose that the decoder stops decoding at time $t= l k$ when it receives the coded packet $S =k$. 
The followings are necessary and sufficient conditions for a decoding period to contain just $k$ decoding packets given that the decoding process starts with $r$ erasures in the previous interval:
\begin{enumerate}
\item[(i)] precisely $k-r$ erasures happen in the period of $lk$; 
\item[(ii)] the patterns of these erasures shall be  \emph{admissible}. The term \emph{admissible} is used here to mean that at least  one erasure should happen during a time $t=(r+1)l$, at least two within $(r+2)l$, and so on, so that the decoding process will remain activated for the whole of the time $t= kl$. 
\end{enumerate} 
The probability of (i) is 
\begin{equation}
\mathcal{B}(k-r; \epsilon, kl) = {kl \choose k-r} \epsilon^{k-r}(1-\epsilon)^{k(l-1)+r}
\end{equation}
and the probability of any pattern of the $k-r$ erasures being admissible is $\nicefrac{r}{k}$.
\end{lemma}





\begin{lemma}\label{sumchoose}
\begin{equation}
\sum_{r=2}^{\min(k,l)}  (r-1) {l \choose r}  {l(k-1) \choose k-r} = {(k-1)l \choose k}
\end{equation}
\end{lemma}
\begin{proof}
We use the following identities that follow from Vandermonde's identity:
\small
\begin{align}
&\sum_{r=2}^{\min(k,l)}   {l \choose r}  {(k-1)l \choose k-r}   = {kl \choose k}  
   - {(k-1)l \choose k}  
   -  l {(k-1)l  \choose k-1} \label{sumchoose1} \\
&\sum_{r=2}^{\min(k,l)}   {l-1 \choose r-1}  {(k-1)l \choose k-r}   = {kl-1 \choose k-1} - {(k-1)l \choose k-1} \label{sumchoose2}\\
&\sum_{r=2}^{\min(k,l)}   {l-1 \choose r}  {(k-1)l \choose k-r}   = {kl-1 \choose k} - {(k-1)l \choose k}  \nonumber\\ 
&\quad - (l-1) {(k-1)l  \choose k-1} \label{sumchoose3}
\end{align}
\normalsize
Hence, 
$\sum_{r=2}^{\min(k,l)}  (r-1) {l \choose r}  {l(k-1) \choose k-r}   = \frac{l}{2} \sum_{r=2}^{\min(k,l)} \bigl[ (1-\frac{2}{l}) 
  +{l-1 \choose r-1} +{l-1 \choose r}\bigr] {l(k-1) \choose k-r}
   \overset{(\ref{sumchoose1}),(\ref{sumchoose2}),(\ref{sumchoose3})}{{=}} {(k-1)l \choose k}$.
\end{proof}

\begin{proof}[Proof of Theorem \ref{maingeniedet}, Part I]
We know that for the decoding process to go beyond $k$ we need at least $k$ erasures in the first $kl$ interval. This means that there are cases with more than $k$ erasures in the first $kl$ interval that the decoding process stops before $k$ but for it to be greater than $k$ we should have at least $k$ erasures. Suppose $E_{kl}$ denotes the number of erasures in the interval $kl$. We have $ P(S>k | E_{kl}\leq k )=0$. Formally speaking:
%
 $P(S > k )  =   P(E_{kl}>k) P(S>k | E_{kl}>k ) 
 +P(E_{kl}< k) P(S>k | E_{kl}<k ) 
 =   P(E_{kl}>k) P(S>k | E_{kl}>k ) 
  <  P(E_{kl}>k) = P(\mathcal{B}(\epsilon, kl )>k) 
  <  Q(\sqrt{k} \frac{1-l\epsilon}{\sqrt{l \epsilon (1-\epsilon)}} ) < \frac{1}{2} e^{-k \frac{(1-l\epsilon)^2}{{l \epsilon (1-\epsilon)}}}$. 
The Q-function is the tail probability of the standard normal distribution,  
$Q(x) = \frac{1}{\sqrt{2\pi}} \int_x^\infty \exp\left(-\frac{u^2}{2}\right) \, du.$
In the last part we used the Chernov bound and the Normal distribution $(\epsilon k l ,  k l \epsilon (1-\epsilon))$ as an approximation of the binomial distribution for large values of $k$. This bound only holds if $l \epsilon <1$. We know $E(S) = \sum_k P(S>k)$. Putting these two together we can deduce that $E(S)$ is finite if $l \epsilon<1$.    
\end{proof}


\begin{proof}[Proof of Theorem \ref{maingeniedet}, Part II]
The goal is to determine $P(S=k)$. We know that $P(S=0)$ if there is no erasure in the first $l$-interval or only the coding packet is erased, so
$P(S=0) = (1-\epsilon)^l + \epsilon (1-\epsilon)^{l-1} =(1-\epsilon)^{l-1}$. 
%
We also know that $P(S=1)$ if one and only one of the information packets is lost in the first $l$-interval, so
$P(S=1) = (l-1) \epsilon (1-\epsilon)^{l-1}$. 
%
We know that for $k>1$ we need at least $r=2$ erasures in the first $l$-interval which has the probability of 
$P(\mathcal{B}(\epsilon, l )=r) = {l \choose r} \epsilon^{r} (1-\epsilon)^{l-r}$ for $2 \leq r \leq l$. 
Starting from the second decoding $l$-interval, there are $r-1$ erasures to be taken care of. The reason is that either the first coded packet is erased, which means $r-1$ information packets are erased, or the first coded packet is not erased and it will eventually help us to decode one of the $r$ erasures. 
Using Lemma \ref{lemmaTanner}, knowing that there has been $r$ erasures in the first $l$-interval, we have the following for $k>1$,
$P(S=k | r ) = \frac{r-1}{k-1} P(\mathcal{B}(\epsilon, l(k-1) )= k-r)$.
This means that we can allow $k-r$ erasures in the $l(k-1)$-interval. 
Putting these two together we have

\small
\begin{align*}
&P(S=k  ) 
 = \sum_{r=2}^{\min(k,l)} P(\mathcal{B}(\epsilon, l )=r)  \frac{r-1}{k-1} P(\mathcal{B}(\epsilon, l(k-1) )= k-r) \\
& = \sum_{r=2}^{\min(k,l)}  {l \choose r} \epsilon^{r} (1-\epsilon)^{l-r}  \frac{r-1}{k-1} 
  {l(k-1) \choose k-r} \epsilon^{k-r} (1-\epsilon)^{k(l-1)-l+r} \\
& = \sum_{r=2}^{\min(k,l)}  \frac{r-1}{k-1} {l \choose r}  {l(k-1) \choose k-r}   \epsilon^{k} (1-\epsilon)^{k(l-1)} \\
& = \frac{1}{k-1}    \epsilon^{k} (1-\epsilon)^{k(l-1)} \sum_{r=2}^{\min(k,l)}  (r-1) {l \choose r}  {l(k-1) \choose k-r}.
\end{align*}
\normalsize
Using Lemma \ref{sumchoose}, we have
\small
\begin{align*}
&P(S=k  ) 
 = \frac{1}{k-1}    \epsilon^{k} (1-\epsilon)^{k(l-1)} {(k-1)l \choose k} \\
& = \frac{1}{k-1}    \epsilon^{k} (1-\epsilon)^{k(l-1)}  \frac{((k-1)l)!}{ k! ((k-1)l -k )!} \\
& = \frac{((k-1)l -k+1}{(k-1)k}    \epsilon^{k} (1-\epsilon)^{k(l-1)} 
\frac{((k-1)l)!}{ (k-1)! ((k-1)l -k +1 )!}  \\
& = \frac{(k-1)(l-1)}{(k-1)k}    \epsilon^{k} (1-\epsilon)^{k(l-1)} {(k-1)l \choose k-1} \\
& = \frac{l-1}{k}    \epsilon^{k} (1-\epsilon)^{k(l-1)} {(k-1)l \choose k-1}
\end{align*}
\normalsize
\end{proof}

\begin{proof}[Proof of Theorem \ref{maingeniedet}, Part III]
We know that for $l \epsilon <1 $ the mean of probability distribution of $S$ exists, we have
$\sum_k P(S=k) =1$,
which means
$\sum_{k=1}^\infty  \frac{l-1}{k}    \epsilon^{k} (1-\epsilon)^{k(l-1)} {(k-1)l \choose k-1} = 1- (1-\epsilon)^{l-1}$. 
By taking derivatives with respect to $\epsilon$, we have
$(1-\epsilon)^{l-2}=\sum_{k=1}^\infty   \epsilon^{k-1} (1-\epsilon)^{k(l-1)} {(k-1)l \choose k-1} 
-  (l-1) \epsilon^{k} (1-\epsilon)^{k(l-1)-1} {(k-1)l \choose k-1}$.   
That is,
$\frac{E(S)}{(l-1)\epsilon} - \frac{E(S)}{1-\epsilon}   =  (1-\epsilon)^{l-2}$
and by re-arranging, we have
$E(S) = \frac{(l-1) \epsilon (1-\epsilon)^{l-1}}{1-l \epsilon}$.
%
By taking derivatives with respect to $\epsilon$, we have
$\frac{E(S^2)}{\epsilon} - \frac{E(S^2)(l-1)}{1-\epsilon}   =  \frac{\partial}{\partial \epsilon} E(S)$
which means
$E(S^2)  = \frac{(l-1) \epsilon (1- \epsilon)^{l-1} ( (1 - l \epsilon)^2 + l \epsilon (1-\epsilon) )}{(1-l \epsilon)^3} 
 = E(S) + \frac{l(l-1) \epsilon^2 (1- \epsilon)^{l} }{(1-l \epsilon)^3}$.
\end{proof}

%% file: app4.tex
\begin{proof}
Consider a transmission of stream with length $N_t$ a multiple of $l$ at time $t$. This translates to $\frac{N_t}{l} $ l-intervals in time $t$ and assume that the decoding process is consisted of $l$-intervals of length $s_1, s_2, \dots , s_n, \dots$. Remember that $S_1, S_2, \dots , S_n$ is sequence of  positive independent identically distributed random variables. Consider the $S^+$ process, and define $J_n$ as 
$J_n = \sum_{i=1}^n S_i^+, \qquad n>0$
and the renewal interval $[J_n,J_{n+1}]$ is a decoding $l$-interval. 
Then the random variable $(N_t)_{t\geq0}$ given by
$
 X_t =  \sum^{\infty}_{n=1} \mathbb{I}_{\{J_n \leq t\}}=\sup \left\{\, n: J_n \leq t\, \right\}$
(where $\mathbb{I}$ is the indicator function) represents the number of decoding $l$-intervals that have occurred by time $t$, and is a renewal process.

Let $W_1, W_2, \ldots$ be a sequence of $i.i.d.$ random variables denoting the sum of in order delivery delay in each decoding $l$-interval. We have two cases to consider.  Case (i): suppose the $j$'th period is an idle period.  Then $S_j=0$ and the information packets are delivered in-order with no delay.  Case (ii): suppose the $j$'th period is a busy period and the information packet erasure that initiated the busy period started in the first slot $t_{i(j)}+1$.   Then the first information packet is delayed by $S_jl$ slots, the second by $S_jl-1$ slots and so on.   The sum-delay over all of the information packets in the busy period is therefore $\sum_{k=1}^{S_jl}k - \sum_{k=0}^{S_j-1}kl < \frac{S_j^2l(l-1)}{2}$.  

The random variable
$Y_t = \sum_{i=1}^{X_t}W_i $
is a renewal-reward process and its expectation is the sum of in order delivery delay over the time-span of $t$. 
%
Based on the elementary renewal theorem for renewal-reward processes, we have 
$\lim_{t \to \infty} \frac{1}{t} \mathbb{E}[Y_t]   = \frac{\mathbb{E}(W_1)}{\mathbb{E}(S_1^+)}$.
%
Based on the construction of $W_i$, we have
$\mathbb{E}[W_1]  = \sum_{i=1}^{\infty} \mathbb{E}[W_1|S_1^+= i] \Pr(S_1^+=i)
  = \sum_{i=0}^{\infty} \mathbb{E}[W_1|S_1= i] \Pr(S_1=i)
   = \frac{l(l-1)}{2} E[S^2]$. 
We have
$
\lim_{t \to \infty} \frac{1}{t} \mathbb{E}[Y_t] \leq  \frac{1-l \epsilon}{(1-\epsilon)^{l}} \frac{l(l-1)}{2} E[S^2]
$
Since we assumed that $N_t = t l$, we have:
$
\lim_{N_t \to \infty} \frac{1}{N_t} \mathbb{E}[Y_t] \leq  \frac{(l-1)(1-l \epsilon)}{2(1-\epsilon)^{l}}  E[S^2]
$

\end{proof}

%% file: app3.tex
Unfortunately, we cannot extend the Tanner's result to this case in which is equivalent of 
computing the busy time of a $G/D/c$ queue. Instead we use the following lemma :

\begin{lemma}[\cite{ kreweras65}]\label{kreweras} 
Consider the strictly increasing and integer infinite sequence $\mathbf{a} = (a_1 , a_2  , \dots, a_n ,\dots)$.  The number $\mathcal{N}(\mathbf{a},n)$ of sequences of non-negative integers $(b_1, \dots, b_n, \dots)$ dominated by the sequence $(a_1 , a_2 -a_1 , \dots, a_n - a_{n-1},\dots)$ for a given $n$ so that

\small
\begin{align*}
b_1 & \leq a_1\\
b_1 + b_2 & \leq a_2 \\
\ldots & \\
b_1+b_2 + \dots + b_n & \leq a_n  
\end{align*}
\normalsize
is equal to the determinant of the following matrix:
\footnotesize
\[
\left(
\begin{array}{cccccc}
 {a_n +1 \choose 1} & 1  & 0 & \dots &  &  0 \\
  {a_{n-1} +1 \choose 2} &  {a_{n-1} +1 \choose 1}  &  1& 0 & \dots&  0\\
  {a_{n-2} +1 \choose 3} &  {a_{n-2} +1 \choose 2}  & {a_{n-2} +1 \choose 1}& 1 & \dots&  0\\
  \vdots &   & &  & &  \\
    {a_{2} +1 \choose n-1} &  {a_{2} +1 \choose n-2}  & {a_{2} +1 \choose n-3 }& \dots & {a_{2} +1 \choose 1 }&  1 \\
    {a_{1} +1 \choose n} &  {a_{1} +1 \choose n-1}  & {a_{1} +1 \choose n-2 }& \dots & {a_{1} +1 \choose 2 }&   {a_{1} +1 \choose 1 }
     \end{array}
\right)
\]
\normalsize
which can be computed recursively as follows
\[
\mathcal{N}(\mathbf{a},n) = \sum_{j=1}^n (-1)^{j-1} { a_{n-j+1} \choose j} \mathcal{N}(\mathbf{a},n-j).
\]

\end{lemma}

\begin{proof}[Proof of Theorem \ref{maingeniedetgroup}, Part I]
We know that for the decoding process to go beyond $k$ we need at least $ck$ erasures in the first $kl_g$ interval. This means that there are cases with more than $kc$ erasures in the first $kl_g$ interval that the decoding process stops before $k$ but for it to be greater than $k$ we should have at least $ck$ erasures. Suppose $E_{kl_g}$ denotes the number of erasures in the interval $kl_g$. We have $ P(S>k |_g E_{kl_g}\leq k )=0$. Formally speaking 
%
 $P(S > k )  =   P(E_{kl_g}>ck) P(S>k |_g E_{kl_g}>ck ) 
 +P(E_{kl_g}< ck) P(S>k | E_{kl_g}<ck ) 
 =   P(E_{kl_g}>ck) P(S>k |_g E_{kl_g}>ck ) 
  <  P(E_{kl_g}>ck) = P(\mathcal{B}(\epsilon, kl_g )>ck)
  <  Q(\sqrt{k} \frac{c-l_g\epsilon}{\sqrt{l_g \epsilon (1-\epsilon)}} ) < \frac{1}{2} e^{-k \frac{(c-l_g\epsilon)^2}{{l_g \epsilon (1-\epsilon)}}}$.
This bound only holds if $l_g \epsilon <c$. We know $E(S) = \sum_k P(S>k)$. Putting these two together we can deduce that $E(S)$ is finite if $l_g \epsilon<c$.    
\end{proof}


\begin{proof}[Proof of Theorem \ref{maingeniedetgroup}, Part II]
The goal is to determine $P(S=k)$. We know that $P(S=0)$ if there is none of the information packets in the first $l_g$-interval is erased
$P(S=0)  =(1-\epsilon)^{l_g-c}$. 
%
We also know that $P(S=1)$ if at least one information packet has been lost and the total number of erasures in the first $l_g$ interval is at most $c$, so
$P(S=1) =  \sum_{i=1}^c \sum_{j=0}^{c-i}  {l-c \choose i} {c \choose j}   \epsilon^{i+j}  (1-\epsilon)^{l-i-j}$. 
%
We know that for $S=k>1$, we need the number of erasures in the $kl_g$ interval to be between $kc-c+1$ and $kc$. We also need at least $c+1$ erasures in the first $l_g$-interval, at least $2c+1$ erasures in the first $2l_g$-interval and so on, so that the decoding process remains activated.

Suppose that $S=k$  and consequently the number of erasures in the $kl_g$ interval is $kc-p$ for a given $p$, $0\leq p \leq c-1$. In the remainder of this part of the proof, the goal is to find the number $NP(l_g,c,k,p)$ of erasure patterns with $kc-p$ erasures so that the decoding process remains activated till the reception of coded packets in the $k$-th $l_g$-interval. 

We formulate an erasure pattern by a $\{ 0, 1\}$ sequence $\mathbf{x}$ with the length of $k l_g$. We set $x_i=1$ if an erasure occurs at time $i$. We shall describe the positions of the $1$'s in $\mathbf{x}$ by the sequence $\mathbf{u} = (u_1, u_2,\dots , u_{k c - p})$, where the $i$'s erasure is found in position $u_i$ in the sequence $\mathbf{x}$ and denote by  $\mathbf{\mu} = (\mu_1, \mu_2,\dots , \mu_{k c - p})$ the sequence of differences $\mu_i = u_i - u_{i-1}$ (with $\mu_1=u_1$).
Assuming that there exist $kc-p$ erasures and the decoding process is activated till the $k$th $l_g$ interval, we know that there should exist $c+1$ erasures in the first $l_g$ interval. Hence for an arbitrary admissible erasure pattern $\mathbf{x}$, the first erasure can not occur later than the time $l_g-c$, so we have $\mu_1 \leq l_g -c$. Following the same argument, $\mathbf{\mu}$ should satisfy the following inequalities:

\footnotesize
\begin{align*}
\mu_1 & \leq l_g -c\\
\mu_1+\mu_2  & \leq l_g -c+1\\
\vdots & \\
\mu_1 + \mu_2 + \dots + \mu_{c+1} & \leq l_g \\
\mu_1 + \mu_2 + \dots + \mu_{c+1}+ \mu_{c+2} & \leq 2 l_g - c+1 \\
\mu_1 + \mu_2 + \dots + \mu_{c+2}+\mu_{c+3} & \leq 2 l_g - c+2 \\
\vdots & \\
\mu_1 + \mu_2 + \dots + \mu_{2c}+\mu_{2c+1} & \leq 2 l_g  \\
\vdots & \\
\mu_1 + \mu_2 + \dots + \mu_{(k-2)c+1}+ \mu_{(k-2)c+2} & \leq (k-1) l_g - c+1 \\
\mu_1 + \mu_2 + \dots + \mu_{(k-2)c+2}+\mu_{(k-2)c+3} & \leq (k-1) l_g - c+2 \\
\vdots & \\
\mu_1 + \mu_2 + \dots + \mu_{(k-2)c+c}+\mu_{kc-c+1} & \leq (k-1) l_g \\
\vdots & \\
\mu_1 + \mu_2 + \dots + \mu_{kc-c+1}+ \mu_{kc-c+2} & \leq k l_g - c+p + 1\\
\mu_1 + \mu_2 + \dots + \mu_{kc-c+2}+\mu_{kc-c+3} & \leq k l_g - c+ p+2 \\
\vdots & \\
\mu_1 + \mu_2 + \dots + \mu_{kc-p-1}+\mu_{kc-p} & \leq k l_g 
\end{align*}   
\normalsize
Using Lemma \ref{kreweras}, we can enumerate such admissible erasure patterns for given values of $k, l_g , c , p$. 
\end{proof}

%% file: app1a.tex
 \begin{proof}[Proof of Theorem \ref{throughputthm} ]
 Assume that the decoding process consists of intervals of lengths $s_1, s_2, s_3 , \dots , s_k , s_{k+1} , \dots$, and without the loss of generality, assume that
 \small
 \begin{equation}   
 \# \{ s_j = 0 \}_{ j \leq k}+  \sum_{j=1}^{k} s_j < \frac{N}{l} < \# \{ s_j = 0 \}_{j \leq k+1} + \sum_{j=1}^{k+1} s_j.  \label{boundonN} 
 \end{equation}
 \normalsize
 Note that if $S=0$, then no erasure has occurred in the $(l-1)$-interval and all information packets are received without delay.
 The throughput of this scheme is
 $GT = (l-1) ( \# \{ s_j = 0 \}_{j \leq k} +\sum_{j=1}^{k} s_j)/N
 = (l-1) (\# \{ s_j = 0 \}_{j \leq k+1} +\sum_{j=1}^{k+1} s_j 
    -  s_{k+1} + \mathcal{I}(s_{k+1}=0))/N$,
 where $\mathcal{I}(s_{k+1}=0)$ is the indicator function and is equal to one if $s_{k+1}=0$ and otherwise it is zero.
 From (\ref{boundonN}), we know 
%
From theorem \ref{maingeniedet}, we know the probability distribution of $S_{k+1}$, so
$\Pr{( GT  > \frac{l-1}{l} - \frac{(l-1)  f } {N} )} = \Pr{(S_{k+1} \leq f)}$.
In particular, using the Chernov bound for $\Pr(S>f)$, we have
$\Pr{( GT  > \frac{l-1}{l} - \frac{(l-1)  f } {N} )}  > 1- \frac{1}{2} e^{-f \frac{(1-l\epsilon)^2}{{l \epsilon (1-\epsilon)}}}$.
Setting $\delta = \frac{1}{2} e^{-f \frac{(1-l\epsilon)^2}{{l \epsilon (1-\epsilon)}}}$, we can solve for $f_{\delta}$. Assume that $R_0 = \frac{l-1}{l} - \gamma$, for any $N > \frac{(l-1)f_{\delta}}{\gamma}$, we have 
$\Pr(GT > R_0) > 1- \delta$.
  \end{proof}

%% file: app2.tex
\begin{proof}[Proof of Theorem \ref{theoremfailure} ]


Consider an admissible decoding matrix $G_{k \times k}$ for a given erasure pattern $E_1, E_2, \dots , E_k$. Looking at this matrix from a column perspective, the number of zeros in each column $i$, denoted by $Z_i$ is equal to the number of rows $j$ in which $\sum_j E_j$ is strictly less than $i$. Based on this construction, the sequence $Z_i$ is increasing. 
Let $\{ \nu_1,\dots,\nu_k \}$ be the columns of the decoding matrix $G$. The number of zeros $Z_i$ in each vector $\nu_i$ can be computed based on the erasure pattern. The vector $ \nu_i$ is a $k$-dimensional vector and its first $Z_i$ element is forced to be zero and the other $k-Z_i$ are drawn identically and independently from a uniform distribution drawn from an alphabet with size $Q$. 

Our argument in this part of the proof is based on the Landsburg's combinatorial argument in \cite{berlekamp80}. We now, compute the probability that $k$ columns $\{ \nu_1,\dots,\nu_k \}$ of $G$ are linearly independent. Starting from the last vector $\nu_k$, we know that it should be non-zero which happens with probability $\frac{Q^{k - Z_k}-1}{Q^{k-Z_k}} = 1-  \frac{1}{Q^{k-Z_k}} $. The vector $\nu_{k-1}$ must lie outside of the one-dimensional subspace containing $0$ and the first vector $\nu_k$. This happens with probability of $\frac{Q^{k - Z_{k-1}}-Q}{Q^{k-Z_{k_1}}} = 1-  \frac{Q}{Q^{k-Z_{k-1}}}$. The $i$-th vector $\nu_i$ must lie outside the $(k-i)$ -dimensional subspace spanned by the last $k-i$ vectors $\{ \nu_k , \nu_{k-1}, \dots , \nu_{k-i+1} \}$. This happens with probability of $\frac{Q^{k - Z_{i}}-Q^{k-i}}{Q^{k-Z_{i}}} = 1-  \frac{Q^{k-i}}{Q^{k-Z_{i}}}$. Putting all these together, the probability that the decoding matrix $G_{k \times k}$ is of full rank is given by
$\Pr(\textrm{rank}(G) = k )=  \prod_{i=1}^{k}  (1-  \frac{Q^{k-i}}{Q^{k-Z_{i}}} )$.

Considering different erasure patterns, 
$\Pr(\textrm{rank}(G) = k )$ is maximized if all $Z_i$s are zero. This means that all the erasures happen in the first $l$-interval, i.e. $E_1=k, E_2=0, \dots , E_k=0$  and $Z_1=Z_2=\dots=Z_k=0$. In this case the probability of the matrix $G$ to be full rank is given by
$\Pr(\textrm{rank}(G) = k )=  \prod_{i=1}^{k}  (1-  \frac{Q^{k-i}}{Q^{k}} )$.
 %
$\Pr(\textrm{rank}(G) = k )$ is minimized if all $Z_i$s are maximized. Since $Z_i$ is an increasing sequence, in order to maximize $Z_i$ we have to also maximize $Z_{i+1 }, Z_{i+2}, \dots $. The admissible erasure pattern that maximizes the sequence $Z_i$ is $E_1=2, E_2=1, E_3=1, \dots ,E_{k-1}=1 , E_k =0,$ and $Z_1=0, Z_2=0, Z_3=1, Z_4 = 2 , \dots, Z_k= k-2.$  This is justified because if any of the erasures happens sooner the number of rows $j$ in which $\sum_j E_j$ is strictly less than some value of $i$ becomes smaller, contradicting the assumption that the sequence $Z_i$ is maximized. On the other hand if any of the erasures occurs in a later $l$-interval, the erasure pattern will not be admissible. In this case the the probability of the matrix $G$ to be full rank is given by
$\Pr(\textrm{rank}(G) = k ) =  \prod_{i=1}^{2}   (1-  \frac{Q^{k-i}}{Q^{k}} ) . \prod_{i=3}^{k}   (1-  \frac{Q^{k-i}}{Q^{k-(i-2)}} ) 
 =   (1-  \frac{1}{Q} ) (1-  \frac{1}{Q^{2}} )^{k-1} 
 =    \frac{Q}{Q+1} (1-  \frac{1}{Q^{2}} )^{k}$.
\end{proof}

\begin{proof}[Proof of Theorem \ref{theoremfailstr}]
Consider a transmission of stream with length $N_t$ a multiple of $l$ at time $t$. This translates to $\frac{N_t}{l} $ l-intervals in time $t$ and assume that the decoding process is consisted of $l$-intervals of length $s_1, s_2, \dots , s_n, \dots$. Remember that $S_1, S_2, \dots , S_n$ is sequence of  positive independent identically distributed random variables. Consider the $S^+$ process and 
define for each $n>0$ 
$J_n = \sum_{i=1}^n S_i^+$, 
and the renewal interval $[J_n,J_{n+1}]$ is a decoding $l$-interval. 
Then the random variable $(N_t)_{t\geq0}$ given by
$
 X_t =  \sum^{\infty}_{n=1} \mathbb{I}_{\{J_n \leq t\}}=\sup \left\{\, n: J_n \leq t\, \right\}$
(where $\mathbb{I}$ is the indicator function) represents the number of decoding $l$-intervals that have occurred by time $t$, and is a renewal process.
Let $W_1, W_2, \ldots$ be a sequence of $i.i.d.$ random variables denoting the occurrence of the event of decoding failure in each decoding $l$-interval. We define $ W_i$ to be $
\mathbb{I}(\text{Decoding failure occurring in the  $i$-th decoding $l$-interval })$, where $\mathbb{I}$ is the indicator function. Note that $\mathbb{E}(W_i | S_i^+ =k ) = \Pr{(W_i = 1| S_i^+ =k )}$.
Then the random variable
$Y_t = \sum_{i=1}^{X_t}W_i $
is a renewal-reward process and its expectation is the decoding failure probability over the time-span of $t$. 
Based on the elementary renewal theorem for renewal-reward processes \cite{gallager13}, we have:
\begin{equation}\label{eqn:reward}
\lim_{t \to \infty} \frac{1}{t} \mathbb{E}[Y_t]   = \frac{\mathbb{E}(W_1)}{\mathbb{E}(S_1^+)}.
\end{equation}
Based on the construction of $W_i$, we have
\small
\begin{align}
\mathbb{E}[W_1]  &= \sum_{i=1}^{\infty} \mathbb{E}[W_1|S_1^+= i] \Pr(S_1^+=i)\nonumber \\
 & = \sum_{i=0}^{\infty} \mathbb{E}[W_1|S_1= i] \Pr(S_1=i)\nonumber \\
 &\leq  \sum_{i=1}^{\infty} (1-\frac{Q}{Q+1} (1-  \frac{1}{Q^{2}} )^{k}) \Pr(S_1=i).\label{eqn:EW}
\end{align}
\normalsize
The last step follows from the fact that  the decoding failure probability is zero if $S_1 =0$ and for all $S_1 \geq 1$ the decoding failure probability is upper bounded by $(1-  \frac{1}{Q^{2}} )^{k} $ following the theorem \ref{theoremfailure}. 
Putting together equations (\ref{eqn:reward}),(\ref{eqn:EW}) and Corollary \ref{maingeniecorr}, we have
$\lim_{t \to \infty} \frac{1}{t} \mathbb{E}[Y_t] \leq  \frac{1-l \epsilon}{(1-\epsilon)^{l}} 
 \sum_{i=1}^{\infty} (1-\frac{Q}{Q+1} (1-  \frac{1}{Q^{2}} )^{k}) \Pr(S_1=i)$ 
Since we assumed that $N_t = t l$, we have:
$\lim_{N_t \to \infty} \frac{1}{N_t} \mathbb{E}[Y_t] \leq  \frac{(1-l \epsilon)}{l(1-\epsilon)^{l}} 
 \sum_{i=1}^{\infty} (1-\frac{Q}{Q+1} (1-  \frac{1}{Q^{2}} )^{k}) \Pr(S_1=i)$. 
\end{proof}

\begin{proof}[Proof of Corollary \ref{corollaryfailure}]
Consider the function $f(\epsilon) = \epsilon (1-\epsilon)^{l-1}$ and its derivative $f'(\epsilon) = (1-l\epsilon) (1-\epsilon)^{l-2}$. $f'(\epsilon)$ is positive for $ \epsilon < \frac{1}{l}$, which means $f(\epsilon)$ is an increasing function. This guarantees that $\epsilon_0$ the solution to the equation $\epsilon_0 (1- \epsilon_0)^{l-1}=(1-\frac{1}{Q^2}) \epsilon (1- \epsilon)^{l-1} $ exists and $\epsilon_0 < \epsilon < \frac{1}{l}$. 
From theorem \ref{maingeniedet}, we have
%
$\sum_{k=1}^\infty  \frac{l-1}{k}    f^k(\epsilon) {(k-1)l \choose k-1} = 1- (1-\epsilon)^{l-1}$ and 
$\sum_{k=1}^\infty  \frac{l-1}{k}    f^k(\epsilon_0) {(k-1)l \choose k-1}  = 1- (1-\epsilon_0)^{l-1}$. 
%
At this point, setting $\alpha  = \frac{1-l \epsilon}{(1-\epsilon)^{l}} $ and following the theorem \ref{theoremfailstr}, we can compute the decoding failure probability:
%
$ \lim_{N_t \to \infty} \frac{1}{N_t} \Pr (\textrm{DF}) 
  \leq  \alpha  \sum_{k=1}^{\infty}(1- \frac{Q}{Q+1}(1-\frac{1}{Q^2})^{k} )  \Pr{(S=k)}
  = \alpha  \sum_{k=1}^{\infty} \Pr{(S=k)}  
-  \alpha \frac{Q}{Q+1} \sum_{k=1}^{\infty} (1-\frac{1}{Q^2})^{k}   \Pr{(S=k)} 
  = \alpha  \sum_{k=1}^{\infty}\frac{l-1}{k}    f^k(\epsilon) {(k-1)l \choose k-1}  
-  \alpha\frac{Q}{Q+1} \sum_{k=1}^{\infty} \frac{l-1}{k}    f^k(\epsilon_0) {(k-1)l \choose k-1} 
 = \alpha (1- (1-\epsilon)^{l-1} - \frac{Q}{Q+1}  (1- (1-\epsilon_0)^{l-1})) 
  =  \alpha( \frac{Q}{Q+1}   (1-\epsilon_0)^{l-1} -   (1-\epsilon)^{l-1}  + \frac{1}{Q+1} )$.

\end{proof}